%% file: arXiv.tex
\documentclass[submission,copyright,creativecommons]{eptcs}
\providecommand{\event}{LSFA 2025} 

\usepackage{iftex}

\ifpdf
  \usepackage{underscore}         
  \usepackage[T1]{fontenc}        
\else
  \usepackage{breakurl}           
\fi

\usepackage{enumitem}
\usepackage[dvipsnames]{xcolor}
\usepackage{comment}
\usepackage{graphicx}
\usepackage{amsthm} 
\usepackage{tikz}
\usetikzlibrary{arrows,decorations,shapes,automata,positioning,decorations.pathmorphing}
\tikzset{
	every picture/.style = {
		thick,
		>=stealth',
		node distance = 1.5em and 3em,
	}
	,
	cross line/.style = {
		preaction = {
			draw=white,
			-,
			line width=4pt
		}
	}
	,
	state/.style = {
		rectangle,
		rounded corners = 5pt,
		font = \footnotesize,
		draw,
		minimum width = 1em,
		minimum height = 1em
	}
	, 
	label-state/.style = {
		sloped,
		font = \scriptsize,
		label distance = -2pt
	}
	, 
	label-edge/.style = {
		font = \scriptsize,    
		label distance = -2pt
	}
}

\usepackage{amsmath}
\usepackage{bold-extra}
\input{macros}

\title{Sequent Calculi for Data-Aware Modal Logics}

\author{Carlos Areces
	\institute{Universidad
	Nacional de C{\'o}rdoba} 
	\institute{and CONICET, Argentina}
	\email{carlos.areces@unc.edu.ar}
\and
Valentin Cassano
	\institute{Universidad
	Nacional de R\'io Cuarto}
    \institute{and CONICET, Argentina}
	\email{valentin@dc.exa.unrc.edu.ar}
\and	
Danae Dutto
	\institute{Universidad Nacional de
	C{\'o}rdoba}
	\institute{and CONICET, Argentina} 
	\email{ddutto@dc.exa.unrc.edu.ar}
\and
Raul Fervari
	\institute{Universidad
	Nacional de C{\'o}rdoba} 
	\institute{and CONICET, Argentina} 
	\email{rfervari@unc.edu.ar}
}

\def\titlerunning{Sequent Calculi for Data-Aware Modal Logics}
\def\authorrunning{C.\ Areces, V.\ Cassano, D.\ Dutto and R.\ Fervari}
\begin{document}
\maketitle

\begin{abstract}
	This document serves as a companion to the paper of the same title, wherein we introduce a Gentzen-style sequent calculus for $\hxpd$.
	It provides full technical details and proofs from the main paper.
	As such, it is intended as a reference for readers seeking a deeper understanding of the formal results, including soundness, completeness, invertibility, and cut elimination for the calculus.
\end{abstract}

\input{arXiv-xpath}

\input{arXiv-calculus}

\input{arXiv-soundness}

\input{arXiv-invertibility}

\input{arXiv-derived}

\input{arXiv-completeness}

\input{arXiv-cut}

\bibliographystyle{eptcs}
\bibliography{biblio}




\end{document}

%% file: macros.tex
\usepackage{amsmath}
\usepackage{amssymb}
\usepackage{prftree}
\usepackage{xspace}
\usepackage{
  url,
  colortbl,
  stmaryrd,
  hyperref,
  cleveref,
  xargs,
  booktabs,
  multicol,
  mathtools}

\newtheorem{theorem}{Theorem}
\newtheorem{proposition}[theorem]{Proposition}
\newtheorem{lemma}[theorem]{Lemma}

\newtheorem{definition}[theorem]{Definition}

\theoremstyle{remark}

 \newtheorem{remark}{Remark}

\usepackage{arydshln}
\usepackage{etoolbox}
\newcommand{\zerodisplayskips}{%
  \setlength{\abovedisplayskip}{5pt}%
  \setlength{\belowdisplayskip}{5pt}%
  \setlength{\abovedisplayshortskip}{5pt}%
  \setlength{\belowdisplayshortskip}{5pt}}
\appto{\normalsize}{\zerodisplayskips}
\appto{\small}{\zerodisplayskips}
\appto{\footnotesize}{\zerodisplayskips}

\Crefname{definition}{Def.}{Defs.}
\Crefname{figure}{Fig.}{Figs.}
\Crefname{theorem}{Thm.}{Thms.}
\Crefname{section}{Sec.}{Secs.}
\Crefname{appendix}{Appx.}{Apps.}
\Crefname{remark}{Remark}{Remarks}

\usepackage[colorinlistoftodos,prependcaption,textsize=tiny]{todonotes}

\newcommand{\midd}{\mathrel{\,\mid\,}}

\newcommand{\<}{\langle}
\renewcommand{\>}{\rangle}

\newcommand{\tup}[1]{\langle #1 \rangle}

\DeclarePairedDelimiterX\setof[2]{\{}{\}}{\,#1 \;\delimsize\vert\; #2\,}

\newcommand{\hxpd}{\text{\rm HXPath}_{\rm D}}

\newcommand{\rr}[1]{{\color{WildStrawberry}#1}}
\newcommand{\cc}[1]{{\color{ProcessBlue}#1}}
\newcommand{\pp}[1]{{\color{Mulberry}#1}}

\newcommand{\taag}[3]{@_{#2}\tup{#1}{#3}}

\usepackage{ifthen}
\newcommand{\cmprif}[1]{
    \ifthenelse{\equal{#1}{=}}
    {
      =_{\compc}
    }
    {
      \ifthenelse{\equal{#1}{\neq}}
      {
        \neq_{\compc}
      }
      {
        \mathrel{#1}
      }
    }
}

\newcommand{\tagg}[3]{\tup{{#2}{:} \cmprif{#1} {#3}{:}}}

\newcommand{\hilbert}{\mathbf{H}}
\newcommand{\gentzen}{\mathbf{G}}

\newcommand{\liff}{\leftrightarrow}

\newcommand{\prop}{\mathsf{Prop}}

\newcommand{\nom}{\mathsf{Nom}}
\newcommand{\mo}{\mathsf{Mod}}

\newcommand{\cmp}{\mathsf{Cmp}}

\newcommandx{\dow}[1][1={a}]{\xspace\mathsf{#1}\xspace}
\newcommand{\dowa}{\xspace\mathsf{a}\xspace}

\newcommand{\compc}{\xspace\mathsf{c}\xspace}

\newcommand{\cmpr}{\mathrel{\blacktriangle}}
\newcommand{\cmpd}{\mathrel{\blacktriangledown}}

\newcommandx{\amodel}[1][1={M}]{\mathfrak{#1}}

\newcommand{\nodes}{\mathrm{N}}
\newcommand{\node}{n}
\newcommand{\R}{\mathrm{R}}
\newcommand{\deq}{\approx}
\newcommand{\dneq}{\not\approx}
\newcommand{\V}{\mathrm{V}}

\renewcommand{\iff}{\mathit{iff}}

\newcommand{\Rho}{\mathrm{P}}

\newcommand{\inv}[1]{#1^{\text{-}1}}

\usepackage{tikz}
\usetikzlibrary{arrows,calc,patterns,positioning,shapes}
\usetikzlibrary{decorations.pathmorphing}
\tikzset{
  modal/.style={>=stealth',shorten >=1pt,shorten <=1pt,auto},
  world/.style={circle,draw,minimum size=1.3cm},
  point/.style={circle,draw,fill=black,inner sep=0.2mm},
  reflexive/.style={->,in=120,out=60,loop,looseness=#1},
  reflexive/.default={5},
  reflexive point/.style={->,in=135,out=45,loop,looseness=#1},
  reflexive point/.default={25},
  label-edge/.style = {font = \footnotesize},
}
\tikzset{
  reflexive above/.style={->,loop,in=120,out=60,looseness=#1},
  reflexive above/.default={7},
  reflexive below/.style={->,loop,in=240,out=300,looseness=#1},
  reflexive below/.default={7},
  reflexive left/.style={->,loop,in=150,out=210,looseness=#1},
  reflexive left/.default={7},
  reflexive right/.style={->,loop,in=30,out=330,looseness=#1},
  reflexive right/.default={7}
}

\usepackage{mathrsfs}

\newcommandx{\aderivation}[1][1={D}]{\mathscr{#1}}

\DeclareMathOperator{\size}{size}

\newcommand{\arule}{\mathrm{P}}

\newcommand{\thescalefactor}{.83}

%% file: arXiv-xpath.tex
\section{Hybrid XPath with Data}
\label{sec:xpath2}

We assume $\prop$, $\nom$, $\mo$, and $\cmp$ are pairwise disjoint sets of symbols for propositions, nominals, modalities, and data comparisons.
We assume $\mo$ and $\cmp$ are finite, and $\prop$ and $\nom$ are countably infinite.

\begin{definition}\label[definition]{def:lang}
    The language of $\hxpd$ has \emph{path} expressions (denoted $\alpha$, $\beta$, \dots) and \emph{node} expressions (denoted $\varphi$, $\psi$, \dots), mutually defined by the grammar:
    \begin{align*}
        \alpha, \beta \coloneqq\ &
        \dowa \midd
        i{:} \midd
        \varphi? \midd
        \alpha\beta
        \\
        \varphi, \psi \coloneqq\ &
            p \midd
            i \midd
            \bot \midd
            \varphi \to \psi \midd
            @_i\varphi \midd
            \tup{\dowa}\varphi \midd
            \<\alpha =_{\compc} \beta\> \midd \<\alpha \neq_{\compc} \beta\>,
    \end{align*}
    where 
    $p \in \prop$,
    $i \in \nom$,
    $\dowa \in \mo$, 
    and 
    $\compc\in\cmp$.
    For path expressions, we use $\epsilon \coloneqq \top?$ to indicate the \emph{empty}~path.
    For node expressions, we use standard abbreviations:
        $\top \coloneqq \bot \to \bot$,
        $\lnot \varphi \coloneqq \varphi \to \bot$,
        $\varphi \lor \psi \coloneqq \lnot\varphi \to \psi$,
        $\varphi \land \psi \coloneqq \lnot(\varphi \to \lnot\psi)$, and
        $\varphi \liff \psi \coloneqq (\varphi \to \psi) \land (\psi \to \varphi)$.
    We also abbreviate:
        $\tup{j{:}}\varphi \coloneqq @_j\varphi$,
        $\tup{\psi?}\varphi \coloneqq \psi \land \varphi$,
        $\tup{\alpha\beta}\varphi \coloneqq \tup{\alpha}\tup{\beta}\varphi$, and
        $[\alpha]\varphi \coloneqq \lnot\tup{\alpha}\lnot\varphi$.
    Finally, we abbreviate
        $[\alpha \cmpr \beta] \coloneqq \lnot\tup{\alpha \cmpd \beta}$.
    In the last abbreviation, we use $\cmpr$ when there is no need to distinguish $=_{\compc}$ and $\neq_{\compc}$, and use ${\cmpd}$ to indicate ${\neq_{\compc}}$ if ${\cmpr}$ is ${=_{\compc}}$, and to indicate ${=_{\compc}}$ if ${\cmpr}$ is ${\neq_{\compc}}$.
\end{definition}

Path and node expressions are interpreted over hybrid data models.

\begin{definition}\label[definition]{def:models}
    A \emph{(hybrid data) model} is a tuple 
$\amodel = \tup{  
        \nodes,
        \{\R_{\dowa}\}_{\dowa \in \mo},
        \{\deq_{\compc}\}_{\compc \in \cmp}, 
        g,
        \V}$,
    where
        $\nodes$ is a non-empty set of nodes;
        each $\R_{\dowa}$ is a (binary) \emph{accessibility relation} on $\nodes$;
        each $\deq_{\compc}$ is an equivalence relation on $\nodes$, called a \emph{comparison};
        $g : \nom \to \nodes$ is a \emph{nominal assignment}; and 
        $\V : \prop \to 2^{\nodes}$ is a \emph{valuation}. 
\end{definition}

The satisfiability relation for path and node expressions is as follows.

\begin{definition}\label[definition]{def:semantics} 
    Let $\amodel=\tup{\nodes, \{\R_{\dowa}\}_{\dowa \in \mo},\{\deq_{\compc}\}_{\compc \in \cmp}, g,\V}$ be a model, and let ${\{\node, \node'\} \subseteq \nodes}$. 
    The satisfiability relation $\Vdash$ is given by the following conditions: 
        \[
        \begin{array}{l @{~~}c@{~~} l}
            \amodel,\node,\node' \Vdash \dowa
            & \iff & \node \R_{\dowa} \node' \\
    
            \amodel,\node,\node' \Vdash i{:}
            & \iff &  g(i) = \node' \\
    
            \amodel,\node,\node' \Vdash \varphi?
            & \iff &  \node = \node' \mbox{ and } \amodel,\node \Vdash \varphi \\
    
            \amodel,\node,\node' \Vdash \alpha\beta
            & \iff & \mbox{exists} ~ \node'' \in \nodes \mbox{ s.t. }
                        \amodel,\node,\node'' \Vdash \alpha \mbox{ and } \amodel,\node'',\node'\Vdash \beta \\
                          	
            \amodel,\node \Vdash p
            & \iff & \node \in \V(p) \\
            
            \amodel,\node \Vdash i
            & \iff & g(i) = \node  \\

            \amodel,\node \Vdash \bot
            &\multicolumn{2}{@{}l}{\text{never}} \\
            
            \amodel,\node \Vdash \varphi\to\psi
            & \iff & \amodel,\node \Vdash \varphi \mbox{ implies } \amodel,\node \Vdash \psi \\

            \amodel,\node \Vdash @_i \varphi
            & \iff & \amodel, g(i)\Vdash \varphi\\


            \amodel,\node \Vdash \tup{\dowa}\varphi 
            & \iff & \mbox{exists $\node'\in\nodes$ s.t.\ } \amodel,\node,\node'\Vdash \dowa \mbox{ and } \amodel,\node'\Vdash \varphi \\ 

            \amodel,\node \Vdash \<\alpha=_{\compc}\beta\>
            & \iff & \mbox{exists $\node',\node''\in \nodes$ s.t.\ } 
                        \amodel,\node,\node' \Vdash \alpha, \ 
                        \amodel,\node,\node'' \Vdash \beta \mbox{ and }
                        \node' \deq_{\compc} \node''
        \\
            \amodel,\node\Vdash\<\alpha\neq_{\compc}\beta\>
            & \iff & \mbox{exists $\node',\node''\in \nodes$ s.t.\ } 
                        \amodel,\node,\node' \Vdash \alpha, \ 
                        \amodel,\node,\node'' \Vdash \beta \mbox{ and } \node' \dneq_{\compc} \node''.\\
        \end{array}
        \]
    Let $\Psi$ be a set of node expressions, we use $\amodel,\node\Vdash\Psi$ to indicate $\amodel,\node\Vdash\psi$ for all $\psi\in\Psi$. 
    We say $\Psi$ is \emph{satisfiable} iff there exists $\amodel,\node$ s.t.\ $\amodel,\node \Vdash \Psi$.
    We call a node expression $\varphi$ a \emph{consequence} of $\Psi$, written $\Psi \vDash \varphi$, iff $\Psi\cup\{\lnot \varphi\}$ is unsatisfiable.
    If $\Psi$ is the empty set, we write $\vDash \varphi$ and call $\varphi$ a \emph{tautology}.
\end{definition}

\Cref{prop:abbrv}, immediate from \Cref{def:semantics}, establishes that the abbreviations have their intended meaning.

\begin{proposition}\label[proposition]{prop:abbrv}
$\amodel,\node \Vdash \<\alpha\>\varphi$ iff $\mbox{exists $\node' \in \nodes$ s.t.},~
            \amodel,\node,\node' \Vdash\alpha$ and  
            $\amodel,\node' \Vdash \varphi$.
    Moreover, $\amodel,\node \Vdash [\alpha=_{\compc}\beta]$ iff for all $\node',\node''\in \nodes$,  
                $\amodel,\node,\node' \Vdash \alpha$,  
                $\amodel,\node,\node'' \Vdash \beta$ implies 
                $\node' \deq_{\compc} \node''$.
    Finally,
    $\amodel,\node\Vdash[\alpha\neq_{\compc}\beta]$ iff for all $\node',\node''\in \nodes$, 
                $\amodel,\node,\node' \Vdash \alpha$,  
                $\amodel,\node,\node'' \Vdash \beta$ implies 
                $\node' \dneq_{\compc} \node''$.
\end{proposition}

We conclude this section with presenting the Hilbert-style axiom $\hilbert$ system for $\hxpd$ in~\cite{ArecesF21}.

\input{hilbert}

\begin{definition}
    The axioms schemas and rules of $\hilbert$ are summarized in \Cref{tab:axioms}.
    The notion of a \emph{deduction} of a node expression $\varphi$ in $\hilbert$ is defined as a finite sequence $\psi_1 \dots \psi_n$ of node expressions such that $\psi_n = \varphi$, and for all
    $1 \leq z < n$, $\psi_z$ is either an instantiation of an axiom schema, or it is obtained via (MP), (Nec), (Name) or (Paste).
    In the cases (Name) and (Paste) it is understood that the nominals used meet the side conditions of the rule.
    We write $\vdash_{\hilbert} \varphi$, and call~\emph{$\varphi$ a theorem (of $\hilbert$)}, if there exists a deduction of $\varphi$ in $\hilbert$. 
    For a set of node expressions $\Psi$, we write $\Psi\vdash_{\hilbert} \varphi$, and say that \emph{$\varphi$ is deducible from $\Gamma$ (in $\hilbert$)}, iff there exists a finite set $\{\psi_1 \dots \psi_m\} \subseteq \Psi$ such that $\vdash_{\hilbert} (\psi_1 \land \dots \land \psi_m) \to \varphi$. We use the symbol $\vdash_{\hilbert}^n$ to indicate that the length of the corresponding deduction is $n$. 
\end{definition} 

\begin{remark}
    We note ($@$-def) and ($\tup{\dowa}$-def) are not part of the axioms for $\hxpd$ in~\cite{ArecesF21}. We include them as axioms since we treat $@$ and $\tup{\dowa}$ as primitive modalities, whereas they are abbreviations in~\cite{ArecesF21}.
\end{remark}

The following result is immediate from~\cite{ArecesF21}.
 
\begin{theorem}[Soundness and Completeness~\cite{ArecesF21}]\label[theorem]{th:hilbert-completeness}
        $\Psi \vdash_{\hilbert} \varphi$ iff $\Psi \vDash \varphi$. 
\end{theorem}

%% file: hilbert.tex
\begin{figure}[t]
  \begin{center}
    \scalebox{\thescalefactor}{
    \begin{tabular}{@{}p{.57\textwidth}@{}@{}p{.51\textwidth}@{}}
      \toprule
      \begin{tabular}[t]{@{}ll@{}}
        \multicolumn{2}{l}{\textbf{\textsc{Basic Axioms}}}
        \tabularnewline
        \midrule
        (CPL) & all tautologies of CPL
        \tabularnewline
        ($@$-def) & $@_i\varphi \liff \<i{:}\varphi? =_{\compc} i{:}\varphi?\>$
        \tabularnewline
        ($\<\dowa\>$-def) & $\<\alpha\>\varphi \liff \<\alpha\varphi? =_{\compc} \alpha\varphi?\>$
        \tabularnewline
        (K) & $[\alpha](\varphi \to \psi) \to ([\alpha]\varphi \to [\alpha]\psi)$
        \tabularnewline
        ($@$K) & $@_i(\varphi \to \psi) \to (@_i\varphi \to @_i\psi)$
        \tabularnewline
        ($@$-self-dual) & $\lnot@_i\varphi \liff @_i\lnot\varphi$
        \tabularnewline
        ($@$-intro) & $i \to (\varphi \liff @_i\varphi)$
        \tabularnewline
        ($@$-refl) & $@_ii$
      \end{tabular}
      &
      \begin{tabular}[t]{@{}ll@{}}
        \multicolumn{2}{l}{\textbf{\textsc{Path Axioms}}}
        \tabularnewline
        \midrule
        (comp-assoc) & $\<(\alpha\beta)\gamma \cmpr \eta\> \liff \<\alpha(\beta\gamma) \cmpr \eta\>$
        \tabularnewline
        (comp-neutral)$^\dagger$ & $\<\alpha\epsilon\beta \cmpr \gamma\> \liff \<\alpha\beta \cmpr \gamma\>$
        \tabularnewline
        (comp-dist) & $\<\alpha\beta\>\varphi \liff \<\alpha\>\<\beta\>\varphi$\\
        \tabularnewline[.5em]

      \multicolumn{2}{l}{\footnotesize $^\dagger$~$\alpha$ or $\beta$ may be omitted (but not both)}
      \end{tabular}
      \tabularnewline
      \midrule
      \begin{tabular}[t]{@{}ll@{}}
        \multicolumn{2}{l}{\textbf{\textsc{Data Comparison Axioms}}}
        \tabularnewline
        \midrule
        (equal) & $\<\epsilon =_{\compc} \epsilon\>$
        \tabularnewline
        ($\cmpr$-comm) & $\<\alpha \cmpr \beta\> \liff \<\beta \cmpr \alpha\>$
        \tabularnewline
        ($\epsilon$-trans) & $(\<\alpha =_{\compc} \epsilon\> \land \<\epsilon =_{\compc} \beta\>) \to \<\alpha =_{\compc} \beta\>$
        \tabularnewline
        (distinct) & $\lnot\<\epsilon \neq_{\compc} \epsilon\>$
        \tabularnewline
        ($@$-data) & $\lnot\<i{:} =_{\compc} j{:}\> \liff \<i{:} \neq_{\compc} j{:}\>$
        \tabularnewline
        (subpath) & $\<\alpha\beta \cmpr \gamma\> \to \<\alpha\>\top$
        \tabularnewline
        ($@{\cmpr}$-dist) & $\<i{:}\alpha \cmpr i{:}\beta\> \liff @_i\<\alpha \cmpr \beta\>$
        \tabularnewline
        ($\cmpr$-test) & $\<\varphi?\alpha \cmpr \beta\> \liff (\varphi \land \<\alpha \cmpr \beta\>)$
        \tabularnewline
        (agree) & $\<j{:}i{:}\alpha \cmpr \beta\> \liff \<i{:}\alpha \cmpr \beta\>$
        \tabularnewline
        (back) & $\<\alpha{i}{:}\beta \cmpr \gamma\> \to \<i{:}\beta \cmpr \gamma\>$
        \tabularnewline
        (${\cmpr}$-comp-dist) & $\<\alpha\>\<\beta \cmpr \gamma\> \to \<\alpha\beta \cmpr \alpha\gamma\>$
      \end{tabular}
      &
      \begin{tabular}[t]{l}
        {\textbf{\textsc{Rules of Inference}}}
        \tabularnewline
        \midrule
        \prftree[r]{(MP)}
        {\vdash \varphi}
        {\vdash \varphi \to \psi}
        {\vdash \psi}
        \tabularnewline[5pt]
        \prftree[r]{(Nec)}
        {\vdash \varphi}
        {\vdash [\alpha]\varphi}
        \tabularnewline[10pt]
        \prftree[r]{(Name)$^\ddagger$}
        {\vdash @_i\varphi}
        {\vdash \varphi}
        \tabularnewline[10pt]
        \prftree[r]{(Paste)$^\ddagger$\hspace*{1.25cm}}
        {\vdash @_i\<\dowa\>j \land \<j{:}\alpha \cmpr \beta\> \to \varphi}
        {\vdash \<i{:}\dowa\alpha \cmpr \beta\> \to \varphi}

      \tabularnewline[.5em]
      
{\footnotesize $^\ddagger$~$i$ and $j$ are different and not in $\varphi$, $\alpha$, $\beta$}

      \end{tabular}
      \tabularnewline
      \bottomrule

    \end{tabular}
    }\vspace*{-10pt}
  \end{center}
  \caption{Axioms System $\hilbert$ for $\hxpd$}
  \label{tab:axioms}
\end{figure}

%% file: arXiv-calculus.tex
\section{A Sequent Calculus for $\hxpd$}
\label{sec:arXiv-calculus}

\begin{definition}\label[definition]{def:sequent}
    A \emph{sequent} is a pair $\Gamma \vdash \Delta$, where $\Gamma$ and $\Delta$ are finite, possibly empty sets of node expressions of the form $\tagg{\cmpr}{i}{j}$ or $@_i\varphi$.\footnote{Following standard notation, we will not use curly brackets when writing down sequents.}
    In a sequent $\Gamma \vdash \Delta$, the set $\Gamma$ is called the \emph{antecedent} and the set $\Delta$ is called the \emph{consequent}.
    In turn, a \emph{rule} is a pair $(\Gamma \vdash \Delta, \{\Gamma_1 \vdash \Delta_1, \dots, \Gamma_n \vdash \Delta_n\})$, where $\Gamma \vdash \Delta$ is called the \emph{conclusion}, and $\{\Gamma_1 \vdash \Delta_1, \dots, \Gamma_n \vdash \Delta_n\}$ the set of \emph{premisses} of the rule.  If the set of premisses is empty, we say that the rule is an \emph{axiom}.
    In each rule, certain node expressions in the conclusion sequent are designated as \emph{principal}.
    These are the node expressions the rule acts upon. The remaining node expressions in the sequent are the \emph{context} of the rule, and are carried unchanged across its application. This distinction enables the rule to isolate and analyze the logical structure of the principal formulas while treating the context uniformly.
    The rules of $\gentzen$ are shown in \Cref{rules:hxpd}, in a standard format---i.e., the conclusion is shown under a line, the premisses are above the line in no particular order, the principal node expressions and the context are identified immediately from the presentation.
\end{definition}

\begin{definition}\label[definition]{def:derivation}
    A \emph{derivation} in $\gentzen$ is a sequent-labeled tree constructed according to the rules in \Cref{rules:hxpd}. More precisely,  each non-leaf node in the tree is a sequent that is the conclusion of a rule in $\gentzen$, its immediate successors in the tree (going upwards) are the premisses of that rule. 
    The root of the tree is called the \emph{end-sequent} of the derivation.
    A sequent is \emph{derivable} if it is the end-sequent of some derivation, and it is \emph{provable} if it has a derivation whose leaves are all instances of (Ax)~or~ ($\bot$).
    We use $\Gamma \vdash_{\gentzen} \Delta$ to indicate that $\Gamma \vdash \Delta$ is provable.
A rule is \emph{derived} iff there is a derivation of the conclusion of the rule whose leaves are either axioms or belong to the premisses of the rule.
\end{definition}

\input{therules.tex}

%% file: therules.tex
\begin{figure}[h!]
   \centerline
  {
  \scalebox{\thescalefactor}
  {
  \begin{tabular}{@{}c@{}}
  
  \toprule
  \textbf{\textsc{Propositional Rules}} \\
  \midrule

  \tabularnewline[-7pt]
  \begin{tabular}{@{}c@{}}
    \prfbyaxiom{\footnotesize(Ax)$^{\ddagger}$}{\varphi, \Gamma \vdash \Delta, \varphi} 
    
    ~

    \prfbyaxiom{\footnotesize($\bot$)}{@_i\bot, \Gamma \vdash \Delta}
    
    ~

    \prftree[r]{\footnotesize($\to$L)}
    {\Gamma \vdash \Delta, @_i \varphi}{@_i \psi, \Gamma \vdash \Delta}
    {@_i(\varphi \to \psi), \Gamma \vdash \Delta}
    
    ~
    
    \prftree[r]{\footnotesize($\to$R)}
    {@_i \varphi, \Gamma \vdash \Delta, @_i\psi}
    {\Gamma \vdash \Delta, @_i(\varphi \to \psi)}

    \tabularnewline[5pt]
    {\footnotesize$^{\ddagger}$ $\varphi$ is of the form $@_ip$, $@_ij$, or $\tup{i{:} =_{\compc} j{:}}$}
  \end{tabular}

  \tabularnewline

  \midrule
  \textbf{\textsc{Rules for Nominals}} \\
  \midrule

  \tabularnewline[-7pt]
  \begin{tabular}{@{}c@{}}

    \prftree[r]{\footnotesize($@$T)}
    {@_ii, \Gamma \vdash \Delta}
    {\Gamma \vdash \Delta}
    
    \quad

    \prftree[r]{\footnotesize($@5$)}
    {@_jk, @_ij, @_ik, \Gamma \vdash \Delta}
    {@_ij, @_ik, \Gamma \vdash \Delta}

    \quad

    \prftree[r]{\footnotesize(Nom)$^{\dagger}$}
    {@_ij, \Gamma \vdash \Delta}
    {\Gamma \vdash \Delta}
    
    \tabularnewline[7pt]

    \prftree[r]{\footnotesize(S$_1$)$^{\ddagger}$}
    {@_j\varphi, @_ij, @_i\varphi, \Gamma \vdash \Delta}
    {@_ij, @_i\varphi, \Gamma \vdash \Delta}    

    \quad

    \prftree[r]{\footnotesize(S$_2$)}
    {\taag{\dowa}{i}{k}, @_jk, \taag{\dowa}{i}{j}, \Gamma \vdash \Delta}
    {@_jk, \taag{\dowa}{i}{j}, \Gamma \vdash \Delta}

    \quad

    \prftree[r]{\footnotesize(S$_3$)}
    {\tagg{=}{j}{k}, @_ij, \tagg{=}{i}{k}, \Gamma \vdash \Delta}
    {@_ij, \tagg{=}{i}{k}, \Gamma \vdash \Delta}

    \tabularnewline[5pt]
    {\footnotesize $^{\dagger}$ $j$ is not in the conclusion}
    \quad
    {\footnotesize $^{\ddagger}$ $\varphi$ is of the form $p$, $\bot$, or $\tup{\dowa}k$}
  \end{tabular}

  \tabularnewline

  \midrule
  \textbf{\textsc{Rules for Modalities}} \\
  \midrule

  \tabularnewline[-7pt]
  \begin{tabular}{@{}c@{}}
    \prftree[r]{\footnotesize($@$L)}
    {@_i\varphi, \Gamma \vdash \Delta }
    {@_j@_i\varphi, \Gamma \vdash \Delta}

    \quad

    \prftree[r]{\footnotesize($@$R)}
    {\Gamma \vdash \Delta,@_i\varphi}
    {\Gamma \vdash \Delta,@_j@_i\varphi}
    
    \quad
    
    \prftree[r]{\footnotesize(${\tup{\dowa}}$L)$^{\dagger}$}
    {@_i\tup{\dowa}j, @_j\varphi, \Gamma \vdash \Delta}
    {@_i\tup{\dowa}\varphi, \Gamma \vdash \Delta}
    
    \quad
    
    \prftree[r]{\footnotesize(${\tup{\dowa}}$R)}
    {@_i\tup{\dowa}j, \Gamma \vdash \Delta, @_i\tup{\dowa}\varphi, @_j\varphi}
    {@_i\tup{\dowa}j, \Gamma \vdash \Delta, @_i\tup{\dowa}\varphi}
    
    \tabularnewline[7pt]

    \prftree[r]{\footnotesize(${\tup{\cmpr}}$L)$^{\ddagger}$}
    {@_i\tup{\alpha}j, @_i\tup{\beta}k, \tagg{\cmpr}{j}{k}, \Gamma \vdash \Delta}
    {@_i\tup{\alpha \cmpr \beta}, \Gamma \vdash \Delta}

    \quad

    \prftree[r]{\footnotesize(${\tup{\cmpr}}$R)}
    {@_i\tup{\alpha}j, @_i\tup{\beta}k, \Gamma \vdash \Delta, @_i\tup{\alpha \cmpr \beta},\tagg{\cmpr}{j}{k}}
    {@_i\tup{\alpha}j, @_i\tup{\beta}k, \Gamma \vdash \Delta, @_i\tup{\alpha \cmpr \beta}}
    
    \tabularnewline[5pt]
    {\footnotesize$^{\dagger}$ $j$ is not in the conclusion}
    \qquad
    {\footnotesize $^{\ddagger}$ $j$ and $k$ are different and not in the conclusion}

  \end{tabular}

  \tabularnewline

  \midrule
  \textbf{\textsc{Rules for Data Comparison}} \\
  \midrule

  \tabularnewline[-7pt]
  \begin{tabular}{@{}c@{}}
    {\prftree[r]{\footnotesize(EqT)}
    {\tagg{=}{i}{i}, \Gamma \vdash \Delta}
    {\Gamma \vdash \Delta}}

    \quad

    {\prftree[r]{\footnotesize(Eq5)}
    {\tagg{=}{j}{k}, \tagg{=}{i}{j}, \tagg{=}{i}{k}, \Gamma \vdash \Delta}
    {\tagg{=}{i}{j}, \tagg{=}{i}{k}, \Gamma \vdash \Delta}}

	\quad

    {\prftree[r]{\footnotesize(NEqL)}
    {\Gamma \vdash \Delta, \tagg{=}{i}{j}}
    {\tagg{\neq}{i}{j}, \Gamma \vdash \Delta}}
    
    \quad

    {\prftree[r]{\footnotesize(NEqR)}
    {\tagg{=}{i}{j}, \Gamma \vdash \Delta}
    {\Gamma \vdash \Delta, \tagg{\neq}{i}{j}}}

  \end{tabular}
  
  \tabularnewline

  \midrule
  \textbf{\textsc{Structural Rules}} \\
  \midrule

  \tabularnewline[-7pt]
  \begin{tabular}{@{}c@{}}
    {\prftree[r]{\footnotesize(Cut)}
      {\Gamma \vdash \Delta, \varphi}
      {\varphi, \Gamma' \vdash \Delta'}
      {\Gamma, \Gamma' \vdash \Delta, \Delta'}}
    
    \quad

    {\prftree[r]{\footnotesize(WL)}
      {\phantom{\varphi,}\, \Gamma \vdash \Delta}
      {\varphi, \Gamma \vdash \Delta}}
    
    \quad

    {\prftree[r]{\footnotesize(WR)}
    {\Gamma \vdash \Delta\phantom{, \varphi}}
    {\Gamma \vdash \Delta, \varphi}}
  \end{tabular}
  
  \tabularnewline
  \bottomrule
  \end{tabular}
  }}
  \caption{Sequent Calculus $\gentzen$ for $\hxpd$.}\label[figure]{rules:hxpd}
\end{figure}

%% file: arXiv-soundness.tex
\section{Soundness}
\label{sec:aeXiv-soundness}
\begin{definition}\label{def:seq:validity}
    A sequent $\Gamma \vdash \Delta$ is \emph{valid} iff for all $\amodel$, it follows that $\amodel \Vdash \Gamma$ implies $\amodel \Vdash \psi$ for some $\psi \in \Delta$.
    A rule preserves validity iff the validity of the premisses of the rule implies the validity of the conclusion of the rule.
\end{definition}


\begin{lemma}[Soundness]\label{lemma:soundness}
    Every rule in $\gentzen$ preserves validity.
\end{lemma}
\begin{proof}
    We present a selection of representative cases below. The remaining cases use a similar argument and can be verified by routine inspection.  In all cases below we reason by contradiction.
    \begin{description}
        \item[\textnormal{(Nom)}] 
            Let $\amodel[A]$ be a model s.t.\ $\amodel[A] \Vdash \Gamma$ and $\amodel[A] \nVdash \psi$ for all $\psi \in \Delta$.
            Then, introduce a new nominal $j$ that is not in $\Gamma,\Delta$ and build a model $\amodel[B]$ that is just like $\amodel[A]$ with the exception that $\amodel[B] \Vdash @_ij$.
            We have $\amodel[B] \Vdash @_ij, @_i\varphi, \Gamma$ and $\amodel[B] \nVdash \psi$ for all $\psi \in \Delta$.
            This contradicts the validity of the premiss of the rule.
        
        \item[\textnormal{($\tup{\cmpr}$L)}] We have: (1) $\cmpr$ is $=_{\compc}$, or (2) $\cmpr$ is $\neq_{\compc}$.
            For (1), take any model $\amodel[A]$ s.t.:
                $\amodel[A] \Vdash @_i\<\alpha =_{\compc} \beta\>, \Gamma$, and
                $\amodel[A] \nVdash \psi$ for all $\psi \in \Delta$.
            The semantics of $@_i\<\alpha =_{\compc} \beta\>$ tells us there are $n$ and $n'$ in $\amodel[A]$ s.t.:
                $\amodel[A], g(i), n \Vdash \alpha$,
                $\amodel[A], g(i), n' \Vdash \beta$, and
                $(n,n') \in {\approx_{\compc}}$.
            Then, choose nominals $j$ and $k$ that do not appear in $\Gamma,\Delta,\alpha,\beta$ and build a model $\amodel[B]$ that is identical to $\amodel[A]$ with the exception that $\amodel[B], n \Vdash j$ and $\amodel[B], n' \Vdash k$.
            It is clear that $\amodel[B] \Vdash @_i\<\alpha\>j, @_i\<\beta\>k, \tagg{=}{j}{k}, \Gamma$ and $\amodel[B] \nVdash \psi$ for all $\psi \in \Delta$.
            This contradicts the validity of the premiss of the rule.
            The case for  (2) is similar.
        
        \item[\textnormal{($\tup{\cmpr}$R)}] We have: (1) $\cmpr$ is $=_{\compc}$, or (2) $\cmpr$ is $\neq_{\compc}$.
            For (1), take any model $\amodel[A]$ s.t.:
                $\amodel[A] \Vdash @_i\<\alpha\>j, @_i\<\beta\>k, \Gamma$ and $\amodel[B] \nVdash \psi$ for all $\psi \in \Delta, @_i\<\alpha =_{\cmpr} \beta\>$.
            In particular, $\amodel[A] \nVdash @_i\<\alpha =_{\compc} \beta\>$.
            This means that for all $n$ and $n'$ in $\amodel[A]$, it follows that
                $\amodel[A], g(i), n \Vdash \alpha$,
                $\amodel[A], g(i), n' \Vdash \beta$, and
                $(n,n') \notin {\approx_{\compc}}$.
            This contradicts the validity of premiss of the rule.
            The case for (2) is similar.
            \qedhere
    \end{description}
\end{proof}

Soundness of $\gentzen$ follows from \Cref{lemma:soundness} by induction on the structure of a derivation of a sequent. 

\begin{theorem}[Soundness]\label{th:soundness}
    Every provable sequent in $\gentzen$ is valid, i.e., $\Gamma \vdash_\gentzen \Delta$ implies $\Gamma \vDash \Delta$.
\end{theorem}

%% file: arXiv-invertibility.tex
\section{Invertibility}
\label{sec:arXiv-invertibility}

\begin{definition}\label[definition]{def:invertible}
    A rule~$(\Rho)$ is \emph{invertible} iff there is a derivation of each premiss of the rule whose leaves are either axioms or the conclusion of the rule.
    Any such derivation is called an \emph{inverse} of the rule and is denoted by~$(\inv{\arule})$.
\end{definition}

\begin{theorem}\label{th:invertible}
    Every rule in $\gentzen$ is invertible.
\end{theorem}
\begin{proof}
    The case for the rules
    (@T), 
    (@5),
    (Nom),
    $(\text{S}_1)$,
    $(\text{S}_2)$,
    $(\text{S}_3)$,
    $(\tup{\dowa}\text{R})$,
    $(\tup{\cmpr}\text{R})$,
    (EqT), and
    (Eq5) 
    the result is immediate using weakening.
    \begin{description}
        \item(${\to}\text{L}_{1}$)
        Given a derivation of $@_i(\varphi \to \psi)$, we can build a derivation of $@_i\varphi$  as follows:
        \begin{center}\scalebox{\thescalefactor}
            {
                \prftree[r]{\footnotesize(Cut)}
                {
                    \prftree[r]{\footnotesize($\to$R)}
                    {
                        \prfbyaxiom{\footnotesize(Ax)}
                        {@_i\varphi, \Gamma \vdash \Delta, @_i\varphi, @_i\psi}
                    }
                    {\Gamma \vdash \Delta, @_i\varphi, \cc{@_i(\varphi \to \psi)}}
                }
                {
                    \prfassumption
                    {\cc{@_i(\varphi \to \psi)},\Gamma \vdash \Delta}
                }
                {\Gamma \vdash \Delta, @_i\varphi}
            }
        \end{center}
        \item (${\to}\text{L}_2$)
        Given a derivation of $@_i(\varphi \to \psi)$, we can build a derivation of $@_i\varphi$  as follows:
        \begin{center}\scalebox{\thescalefactor}
            {
                \prftree[r]{\footnotesize(Cut)}
                {
                    \prftree[r]{\footnotesize($\to$R)}
                    {
                        \prfbyaxiom{\footnotesize(Ax)}
                        {@_i\varphi,@_i\psi, \Gamma \vdash \Delta,@_i\psi}
                    }
                    {@_i\psi, \Gamma \vdash \Delta, \cc{@_i(\varphi \to \psi)}}
                }
                {
                    \prfassumption
                    {\cc{@_i(\varphi \to \psi)},\Gamma \vdash \Delta}
                }
                {@_i\psi, \Gamma \vdash \Delta}
            }
        \end{center}
        \item (${\to}\text{R}$)
        Given a derivation of $@_i(\varphi \to \psi)$, we can build a derivation of $@_i\varphi$ as follows:
        \begin{center}\scalebox{\thescalefactor}
            {
                \prftree[r]{\footnotesize(Cut)}
                {
                    \prfassumption
                    {\Gamma \vdash \Delta, \cc{@_i(\varphi \to \psi)}}
                }
                {
                    \prftree[r]{\footnotesize($\to$L)}
                    {
                        \prfbyaxiom{\footnotesize(Ax)}
                        {@_i\varphi, \Gamma \vdash \Delta, @_i\psi, @_i\varphi}
                    }
                    {
                        \prfbyaxiom{\footnotesize(Ax)}
                        {@_i\psi, @_i\varphi, \Gamma \vdash \Delta, @_i\psi}
                    }
                    {\cc{@_i(\varphi \to \psi)},@_i\varphi, \Gamma \vdash \Delta, @_i\psi}
                }
                {@_i\varphi, \Gamma \vdash \Delta, @_i\psi}
            }
        \end{center}
        \item ($@\text{L}$)
        Given a derivation of $@_j@_i\varphi$, we can build a derivation of $@_i\varphi$ as follows:
        \begin{center}\scalebox{\thescalefactor}
            {
                \prftree[r]{\footnotesize(Cut)}
                {
                    \prftree[r]{\footnotesize($@$R)}
                    {
                        \prfbyaxiom{\footnotesize(Ax)}
                        {@_i\varphi, \Gamma \vdash \Delta, @_i\varphi}
                    }
                    {@_i\varphi, \Gamma \vdash \Delta, \cc{@_j@_i\varphi}}
                }
                {
                    \prfassumption
                    {\cc{@_j@_i\varphi}, \Gamma \vdash \Delta}
                }
                {@_i\varphi, \Gamma \vdash \Delta}
            }
        \end{center}
        \item ($@\text{R}$)
        Given a derivation of $@_j@_i\varphi$, we can build a derivation of $@_i\varphi$ as follows:
        \begin{center}\scalebox{\thescalefactor}
            {
                \prftree[r]{\footnotesize(Cut)}
                {
                    \prfassumption
                    {\Gamma \vdash \Delta, \cc{@_j@_i\varphi}}
                }
                {
                    \prftree[r]{\footnotesize($@$L)}
                    {
                        \prfbyaxiom{\footnotesize(Ax)}
                        {@_i\varphi, \Gamma \vdash \Delta, @_i\varphi}
                    }
                    {\cc{@_j@_i\varphi}, \Gamma \vdash \Delta, @_i\varphi}
                }
                {\Gamma \vdash \Delta, @_i\varphi}
            }
        \end{center}
        \item ($\tup{\dowa}\text{L}$)
        Given a derivation of $@_i\tup{\dowa}\varphi$, we can build a derivation of $@_i\tup{\dowa}j,@_j\varphi$ as follows:
        \begin{center}\scalebox{\thescalefactor}
            {
                \prftree[r]{\footnotesize(Cut)}
                {
                    \prftree[r]{\footnotesize($\<\dowa\>$R)}
                    {
                        \prfbyaxiom{\footnotesize(Ax)}
                        {\taag{\dowa}{i}{j}, @_j\varphi, \Gamma \vdash \Delta, @_i\<\dowa\>\varphi, @_j\varphi}
                    }
                    {\taag{\dowa}{i}{j}, @_j\varphi, \Gamma \vdash \Delta, \cc{@_i\<\dowa\>\varphi}}
                }
                {
                    \prfassumption
                    {\cc{@_i\<\dowa\>\varphi}, \Gamma \vdash \Delta}
                }
                {\taag{\dowa}{i}{j}, @_j\varphi, \Gamma \vdash \Delta}
            }
        \end{center}
        \item ($\text{NEqL}$) 
        Given a derivation of $\tagg{\neq}{i}{j}$, we can build a derivation of $\tagg{=}{i}{j}$ as follows:
        \begin{center}\scalebox{\thescalefactor}
            {
                \prftree[r]{\footnotesize(Cut)}
                {
                    \prftree[r]{\footnotesize(NEqR)}
                    {
                        \prfbyaxiom{\footnotesize(Ax)}
                        {\tagg{=}{i}{j}, \Gamma \vdash \Delta, \tagg{=}{i}{j}}
                    }
                    {\Gamma \vdash \Delta, \tagg{=}{i}{j}, \tagg{\neq}{i}{j}}
                }
                {
                    \prfassumption
                    {\tagg{\neq}{i}{j}, \Gamma \vdash \Delta}
                }
                {\Gamma \vdash \Delta, \tagg{=}{i}{j}}
            }
        \end{center}
        \item (${\tup{\cmpr}}\text{L}$)
        Given a derivation of $\tup{\alpha \cmpr\beta}$, we can build a derivation of $@_i\tup{\alpha}j, @_i\tup{\beta}k, \tagg{\cmpr}{j}{k}$ as follows:
        \begin{center}\scalebox{\thescalefactor}
            {
                \prftree[r]{\footnotesize(Cut)}
                {
                    \prftree[r]{\footnotesize($\<\cmpr\>$R)}
                    {
                        \prfbyaxiom{\footnotesize(Ax)}
                        {@_i\tup{\alpha}j, @_i\tup{\beta}k, \tagg{\cmpr}{j}{k}, \Gamma \vdash \Delta, @_i\tup{\alpha \cmpr \beta}, \tagg{\cmpr}{j}{k}}
                    }
                    {@_i\tup{\alpha}j, @_i\tup{\beta}k, \tagg{\cmpr}{j}{k}, \Gamma \vdash \Delta, \cc{@_i\tup{\alpha \cmpr \beta}}}
                }
                {
                    \prfassumption
                    {\cc{@_i\tup{\alpha \cmpr \beta}}, \Gamma \vdash \Delta}
                }
                {@_i\tup{\alpha}j, @_i\tup{\beta}k, \tagg{\cmpr}{j}{k}, \Gamma \vdash \Delta}
            }
        \end{center}
        \item ($\lnot \text{L}$)
        Given a derivation of $@_i\lnot\varphi$, we can build a derivation of $@_i\varphi$ as follows:
        \begin{center}\scalebox{\thescalefactor}
            {
                \prftree[r]{\footnotesize(Cut)}
                {
                    \prftree[r]{\footnotesize($\to$)R}
                    {
                        \prfbyaxiom{\footnotesize(Ax)}
                        {@_i\varphi \vdash @_i\varphi,@_i\bot}
                    }
                    {\vdash @_i\varphi, \cc{@_i \lnot \varphi}}
                }
                {
                    \prfassumption
                    {\cc{@_i \lnot \varphi}, \Gamma \vdash \Delta}
                }
                {\Gamma \vdash \Delta, @_i\varphi}
            }
        \end{center}
        \item ($\lnot \text{R}$)
        Given a derivation of $@_i\lnot\varphi$, we can build a derivation of $@_i\varphi$ as follows:
        \begin{center}\scalebox{\thescalefactor}
            {
                \prftree[r]{\footnotesize(Cut)}
                {
                    \prfassumption
                    {\Gamma \vdash \Delta, \cc{@_i \lnot \varphi}}
                }
                {
                    \prftree[r]{\footnotesize($\to$L)}
                    {
                        \prfbyaxiom{\footnotesize(Ax)}
                        {@_i\varphi \vdash @_i\varphi}
                    }
                    {
                        \prfbyaxiom{\footnotesize($\bot$)}
                        {@_i\bot, @_i\varphi \vdash}
                    }
                    {\cc{@_i\lnot\varphi},@_i\varphi \vdash}
                }
                {@_i\varphi, \Gamma \vdash \Delta}
            }
        \end{center}
        \item ($\land \text{L}$)
        Given a derivation of $@_i(\varphi\land \psi)$, we can build a derivation of $@_i\varphi,@_i\psi$ as follows:
        \begin{center}\scalebox{\thescalefactor}
            {
                \prftree[r]{\footnotesize(Cut)}
                {
                    \prftree[r]{\footnotesize($\land$R)}
                    {
                        \prfbyaxiom{\footnotesize(Ax)}
                        {@_i\varphi, @_i\psi \vdash @_i\varphi}
                    }
                    {
                        \prfbyaxiom{\footnotesize(Ax)}
                        {@_i\varphi, @_i\psi \vdash @_i\psi}
                    }
                    {@_i\varphi, @_i\psi \vdash \cc{@_i(\varphi \land \psi)}}
                }
                {
                    \prfassumption
                    {\cc{@_i(\varphi \land \psi)}, \Gamma \vdash \Delta}
                }
                {@_i\varphi, @_i\psi, \Gamma \vdash \Delta}
            }
        \end{center}
    \end{description}
\end{proof}

%% file: arXiv-derived.tex
\section{Derived Rules}

\begin{definition}\label[definition]{def:derived}
    A rule~$(\Rho)$ is \emph{derived} iff there is a derivation of the conclusion of the rule whose leaves are either axioms or belong to the premiss of the rule.
\end{definition}

The following proofs establish the derivability of the rules in \Cref{drules:hxpd}.

\input{derivedrules}

\begin{itemize}
    \item ($\top$L)
        \begin{center}\scalebox{\thescalefactor}
            {
                \prftree[r]{\footnotesize(Cut)}
                {
                    \prftree[r]{\footnotesize($\to$R)}
                    {
                        \prfbyaxiom{\footnotesize($\bot$)}
                        {@_i\bot, \Gamma \vdash \Delta, @_i\bot}
                    }
                    {\Gamma \vdash \Delta, \cc{@_i\top}}
                }
                {
                    \prfassumption
                    {\cc{@_i\top}, \Gamma \vdash \Delta}
                }
                {\Gamma \vdash \Delta}
            }
        \end{center}
    \item ($\lnot$L)
        \begin{center}\scalebox{\thescalefactor}
            {
                \prftree[r]{\footnotesize($\to$L)}
                {\Gamma \vdash \Delta, @_i\varphi}
                {
                    \prfbyaxiom{\footnotesize($\bot$)}
                    {@_i\bot, \Gamma \vdash \Delta}
                }
                {@_i\lnot\varphi, \Gamma \vdash \Delta}
            }
        \end{center}
    \item ($\lnot$R)
        \begin{center}\scalebox{\thescalefactor}
            {
                \prftree[r]{\footnotesize($\to$R)}
                {
                    \prftree[r]{\footnotesize(WR)}
                    {@_i\varphi, \Gamma \vdash \Delta}
                    {@_i\varphi, \Gamma \vdash \Delta, @_i\bot}
                }
                {\Gamma \vdash \Delta, @_i\lnot\varphi}
            }
        \end{center}
    \item ($\land$L)
        \begin{center}\scalebox{\thescalefactor}
            {
                \prftree[r]{\footnotesize($\to$L)}
                {
                    \prftree[r]{\footnotesize($\to$R)}
                    {
                        \prftree[r]{\footnotesize($\to$R)}
                        {
                            \prftree[r]{\footnotesize(WR)}
                            {@_i\psi,@_i\varphi, \Gamma \vdash \Delta}
                            {@_i\psi,@_i\varphi, \Gamma \vdash \Delta,@_i\bot}
                        }
                        {@_i\varphi, \Gamma \vdash \Delta, @_i(\psi \to \bot)}
                    }
                    {\Gamma \vdash \Delta, @_i(\varphi \to (\psi \to \bot))}
                }
                {
                    \prfbyaxiom{\footnotesize($\bot$)}
                    {@_i\bot, \Gamma \vdash \Delta}
                }
                {@_i(\varphi \land \psi), \Gamma \vdash \Delta}
            }
        \end{center}
    \item ($\land$R)
        \begin{center}\scalebox{\thescalefactor}
            {
                \prftree[r]{\footnotesize($\to$R)}
                {
                    \prftree[r]{\footnotesize(WR,$\to$L)}
                    {
                        \prfassumption
                        {\Gamma \vdash \Delta, @_i\varphi}
                    }
                    {
                        \prftree[r]{\footnotesize($\to$L)}
                        {
                            \prfassumption
                            {\Gamma \vdash \Delta, @_i\psi}
                        }
                        {
                            \prfbyaxiom{\footnotesize($\bot$)}
                            {@_i\bot,\Gamma \vdash \Delta}
                        }
                        {@_i(\psi \to \bot), \Gamma \vdash \Delta}
                    }
                    {@_i(\varphi \to (\psi \to \bot)), \Gamma \vdash \Delta, @_i\bot}
                }
                {\Gamma \vdash \Delta, @_i(\varphi \land \psi)}
            }

        \end{center}
    \item ($\liff$L)
            \begin{center}\scalebox{\thescalefactor}
                {
                    \prftree[r]{\footnotesize($\to$L)}
                    {
                        \prftree[r]{\footnotesize($\to$R)}
                        {
                            \prftree[r]{\footnotesize($\to$R)}
                            {
                                \prftree[r]{\footnotesize(WR, $\to$L)}
                                {
                                    \prftree[r]{\footnotesize($\to$L)}
                                    {
                                        \prfbyaxiom{\footnotesize(Ax)}
                                        {@_i\varphi, \Gamma \vdash \Delta, @_i\varphi}
                                    }
                                    {
                                        \prfassumption
                                        {@_i\psi, @_i\varphi, \Gamma \vdash \Delta}
                                    }
                                    {@_i\varphi,@_i(\varphi \to \psi), \Gamma \vdash \Delta}
                                }
                                {
                                    \prftree[r]{\footnotesize($\to$L)}
                                    {
                                        \prfassumption
                                        {\Gamma \vdash \Delta, @_i\psi, @_i\varphi}
                                    }
                                    {
                                        \prfbyaxiom{\footnotesize(Ax)}
                                        {@_i\psi, \Gamma \vdash \Delta, @_i\psi}
                                    }
                                    {@_i(\varphi \to \psi), \Gamma \vdash \Delta, @_i\psi}
                                }
                                {\rr{@_i(\psi \to \varphi)},@_i(\varphi \to \psi), \Gamma \vdash \Delta, @_i\bot}
                            }
                            {@_i(\varphi \to \psi), \Gamma \vdash \Delta, @_i((\psi \to \varphi)\to \bot)}
                        }
                        {\Gamma \vdash \Delta, @_i((\varphi \to \psi)\to ((\psi \to \varphi)\to \bot))}
                    }
                    {
                       \prfbyaxiom{($\bot$)}
                       {@_i\bot, \Gamma \vdash \Delta}
                    }
                    {@_i(\varphi \liff \psi), \Gamma \vdash \Delta}
                }
            \end{center}

    \item ($\liff$R)
        \begin{center}\scalebox{\thescalefactor}
            {
                \prftree[r]{\footnotesize($\to$R)}
                {
                    \prftree[r]{\footnotesize(WL,$\to$L)}
                    {
                        \prftree[r]{\footnotesize($\to$R)}
                        {@_i\varphi, \Gamma \vdash \Delta, @_i\psi}
                        {\Gamma \vdash \Delta, @_i(\varphi \to \psi)}
                    }
                    {
                        \prftree[r]{\footnotesize($\to$L)}
                        {
                            \prftree[r]{\footnotesize($\to$R)}
                            {@_i\psi, \Gamma \vdash \Delta, @_i\varphi}
                            {\Gamma \vdash \Delta, @_i(\psi \to \varphi)}
                        }
                        {
                            \prfbyaxiom{\footnotesize($\bot$)}
                            {@_i\bot, \Gamma \vdash \Delta}
                        }
                        {@_i((\psi \to \varphi)\to \bot), \Gamma \vdash \Delta}
                    }
                    {@_i((\varphi \to \psi) \to ((\psi \to \varphi)\to \bot)), \Gamma \vdash \Delta, @_i\bot}
                }
                {\Gamma \vdash \Delta, @_i(\varphi \liff \psi)}
            }
        \end{center}
    \item (MP)
        \begin{center}\scalebox{\thescalefactor}
            {
                \prftree[r]{\footnotesize(Cut)}
                {\Gamma \vdash \Delta, \cc{@_i\varphi}}
                {
                    \prftree[r]{\footnotesize($\inv{{\to}\text{R}}$)}
                    {\Gamma' \vdash \Delta', @_i(\varphi \to \psi)}
                    {\cc{@_i\varphi},\Gamma' \vdash \Delta', @_i\psi}
                }
                {\Gamma,\Gamma' \vdash \Delta, \Delta', @_i\psi}
            }
        \end{center}
    \item ($@$B)
        \begin{center}\scalebox{\thescalefactor}
            {
                \prftree[r]{\footnotesize(@T)}
                {
                    \prftree[r]{\footnotesize(@5)}
                    {
                        \prftree[r]{\footnotesize(WL)}
                        {
                            \prfassumption
                            {@_ji, \Gamma \vdash \Delta}
                        }
                        {@_ji, @_ij, @_ii, \Gamma \vdash \Delta}
                    }
                    {@_ij, @_ii, \Gamma \vdash \Delta}
                }
                {@_ij, \Gamma \vdash \Delta}
            }
        \end{center}
    \item (@$\cmpr$L)
        \begin{center}\scalebox{\thescalefactor}
            {
                \prftree[r]{\footnotesize($\tup{\cmpr}$L)}
                {
                    \prftree[r]{\footnotesize($@$L)}
                    {
                        \prftree[r]{\footnotesize($@$B)}
                        {
                            \prftree[r]{\footnotesize(S$_3$)}
                            {
                                \prftree[r]{\footnotesize(WL,S$_3$)}
                                {
                                    \prftree[r]{\footnotesize(WL)}
                                    {\tagg{\cmpr}{i}{j},\Gamma \vdash \Delta}
                                    {\tagg{\cmpr}{i}{j}, @_bj,\tagg{\cmpr}{i}{b},\Gamma \vdash \Delta}
                                }
                                {\tagg{\cmpr}{i}{b},@_ai,@_bj, \tagg{\cmpr}{a}{b}, \Gamma \vdash \Delta}
                            }
                            {@_ai,@_bj, \tagg{\cmpr}{a}{b}, \Gamma \vdash \Delta}
                        }
                        {\rr{@_ia},\rr{@_jb}, \tagg{\cmpr}{a}{b}, \Gamma \vdash \Delta}
                    }
                    {@_k\tup{i{:}}a, @_k\tup{j{:}}b, \tagg{\cmpr}{a}{b},\Gamma \vdash \Delta}
                }
                {@_k\tagg{\cmpr}{i}{j},\Gamma \vdash \Delta}

            }
        \end{center}
    \item (@$\cmpr$R)
        \begin{center}\scalebox{\thescalefactor}
            {
                \prftree[r]{\footnotesize(Cut)}
                {
                    \prftree[r]{\footnotesize(WR)}
                    {\Gamma \vdash \Delta,\tagg{\cmpr}{i}{j}}
                    {\Gamma \vdash \Delta,@_k\tagg{\cmpr}{i}{j}, \cc{\tagg{\cmpr}{i}{j}}}
                }
                {
                    \prftree[r]{\footnotesize{($@$T)}}
                    {
                        \prftree[r]{\footnotesize($\inv{@\text{L}}$)}
                        {
                            \prftree[r]{\footnotesize($\tup{\cmpr}$R)}
                            {
                                \prfbyaxiom{\footnotesize(Ax)}
                                {@_k\tup{i{:}}i, @_k\tup{j{:}}j,\tagg{\cmpr}{i}{j},\Gamma \vdash \Delta,@_k\tagg{\cmpr}{i}{j}, \tagg{\cmpr}{i}{j}}
                            }
                            {@_k\tup{i{:}}i, @_k\tup{j{:}}j,\tagg{\cmpr}{i}{j},\Gamma \vdash \Delta,@_k\tagg{\cmpr}{i}{j}}
                        }
                        {@_ii, @_jj, \tagg{\cmpr}{i}{j},\Gamma \vdash \Delta,@_k\tagg{\cmpr}{i}{j}}
                    }
                    {\cc{\tagg{\cmpr}{i}{j}},\Gamma \vdash \Delta,@_k\tagg{\cmpr}{i}{j}}
                }
                {\Gamma \vdash \Delta,@_k\tagg{\cmpr}{i}{j}}
            }
        \end{center}
    \item ($\tup{\cmpr}$B)
        \begin{itemize}
            \item $\cmpr$ is $=_{\compc}$
            \begin{center}\scalebox{\thescalefactor}
                {
                    \prftree[r]{\footnotesize(EqT)}
                    {
                        \prftree[r]{\footnotesize(Eq5)}
                        {
                            \prftree[r]{\footnotesize(WL)}
                            {
                                \prfassumption
                                {\tup{j{:} =_{\compc} i{:}}, \Gamma \vdash \Delta}
                            }
                            {\tup{j{:} =_{\compc} i{:}}, \tup{i{:} =_{\compc} j{:}}, \tup{i{:} =_{\compc} i{:}}, \Gamma \vdash \Delta}
                        }
                        {\tup{i{:} =_{\compc} j{:}}, \tup{i{:} =_{\compc} i{:}}, \Gamma \vdash \Delta}
                    }
                    {\tup{i{:} =_{\compc} j{:}}, \Gamma \vdash \Delta}
    
                }
            \end{center}
            \item $\cmpr$ is $\neq_{\compc}$
            \begin{center}\scalebox{\thescalefactor}
                {    
                    \prftree[r]{\footnotesize(NEqL)}
                    {
                        \prftree[r]{\footnotesize(Cut)}
                        {
                            \prftree[r]{\footnotesize($\inv{\text{NeqL}}$)}
                            {\tup{j{:} \neq_{\compc} i{:}}, \Gamma \vdash \Delta}
                            {\Gamma \vdash \Delta, \tup{j{:} =_{\compc} i{:}}}
                        }
                        {
                            \prftree[r]{\footnotesize($\tup{=_{\compc}}$B)}
                            {
                                \prfbyaxiom{\footnotesize(Ax)}
                                {\tup{i{:} =_{\compc} j{:}}, \Gamma \vdash \Delta, \tup{i{:} =_{\compc} j{:}}}
                            }
                            {\tup{j{:} =_{\compc} i{:}}, \Gamma \vdash \Delta, \tup{i{:} =_{\compc} j{:}}}
                        }
                        {\Gamma \vdash \Delta, \tup{i{:} =_{\compc} j{:}}}
                    }
                    {\tup{i{:} \neq_{\compc} j{:}}, \Gamma \vdash \Delta}
                }
            \end{center}
        \end{itemize}
\end{itemize}
\input{alpha-sigma}
\input{prf-ax-s1}

%% file: derivedrules.tex
\begin{figure}[p]
     \centerline
    {
    \scalebox{\thescalefactor}
    {
    \begin{tabular}{@{}c@{}}
    
      \toprule
      \textbf{\textsc{Derived Propositional Rules}} \\
      \midrule
    
      \tabularnewline[-7pt]
      \begin{tabular}{@{}c@{}}
        {\prfbyaxiom{\small(Ax)}
        {@_i\varphi,\Gamma \vdash \Delta,@_i\varphi}}

        ~

        {\prftree[r]{\small($\top$L)}
        {@_i\top, \Gamma \vdash \Delta}
        {\Gamma \vdash \Delta}}
        
        \tabularnewline[10pt]

        {\prftree[r]{\small($\lnot$L)}
        {\Gamma \vdash \Delta, @_i\varphi}
        {@_i\lnot\varphi, \Gamma \vdash \Delta}}
        
        ~

        {\prftree[r]{\small($\lnot$R)}
        {@_i\varphi, \Gamma \vdash \Delta}
        {\Gamma \vdash \Delta, @_i\lnot\varphi}}

        \tabularnewline[10pt]

        \prftree[r]{\small($\land$L)}
        {@_i \varphi, @_i \psi, \Gamma \vdash \Delta }
        {@_i(\varphi \land \psi), \Gamma \vdash \Delta}

        ~
        \prftree[r]{\small($\land$R)}
        {\Gamma \vdash \Delta,@_i \varphi}
        {\Gamma \vdash \Delta, @_i\psi}
        {\Gamma \vdash \Delta,@_i(\varphi \land \psi)}

        \tabularnewline[10pt]

        {\prftree[r]{\small($\liff$L)}
        {@_i\varphi, @_i\psi,\Gamma \vdash \Delta}
        {\Gamma \vdash \Delta, @_i\varphi,@_i\psi}
        {@_i(\varphi \liff \psi),\Gamma \vdash \Delta}}

        ~

        {\prftree[r]{\small($\liff$R)}
        {@_i\varphi, \Gamma \vdash \Delta, @_i\psi}
        {@_i\psi, \Gamma \vdash \Delta, @_i\varphi}
        {\Gamma \vdash \Delta, @_i(\varphi \liff \psi)}}

        \tabularnewline[10pt]

        {\prftree[r]{\small(MP)}
        {\Gamma \vdash \Delta, @_i\varphi}
        {\Gamma' \vdash \Delta', @_i(\varphi \to \psi)}
        {\Gamma, \Gamma' \vdash \Delta, \Delta', @_i\psi}}
    \end{tabular}
    
    \tabularnewline

    \midrule
    \textbf{\textsc{Derived Rules for Modalities}} \\
    \midrule
  
    \tabularnewline[-7pt]
    \begin{tabular}{@{}c@{}}
      {\prftree[r]{\footnotesize(${\tup{\alpha}}$L)$^{\dagger}$}
      {\taag{\alpha}{i}{j}, @_j\varphi, \Gamma \vdash \Delta}
      {@_i\tup{\alpha}\varphi, \Gamma \vdash \Delta}}
      
      \qquad
      
      {\prftree[r]{\footnotesize(${\tup{\alpha}}$R)}
      {\taag{\alpha}{i}{j}, \Gamma \vdash \Delta, @_i\tup{\alpha}\varphi, @_j\varphi}
      {\taag{\alpha}{i}{j}, \Gamma \vdash \Delta, @_i\tup{\alpha}\varphi}}
      
      \tabularnewline[5pt]
      {\footnotesize$^{\dagger}$ j does not occur in the conclusion}
    \end{tabular}
  
    \tabularnewline
  
    \midrule
    \textbf{\textsc{Derived Rules for Nominals}} \\
    \midrule
  
    \tabularnewline[-7pt]

    \prftree[r]{\footnotesize(S$_1$)}
    {@_j\varphi, @_ij, @_i\varphi, \Gamma \vdash \Delta}
    {@_ij, @_i\varphi, \Gamma \vdash \Delta}    

    \quad

    \prftree[r]{\footnotesize(S$_2$)}
    {@_i\tup{\alpha}k, @_jk, @_i\tup{\alpha}j, \Gamma \vdash \Delta}
    {@_jk, @_i\tup{\alpha}j, \Gamma \vdash \Delta}

    \quad

    {\prftree[r]{\footnotesize(S$_3$)}
    {\tagg{\cmpr}{j}{k}, @_ij, \tagg{\cmpr}{i}{k}, \Gamma \vdash \Delta}
    {@_ij, \tagg{\cmpr}{i}{k}, \Gamma \vdash \Delta}}

    \tabularnewline[10pt]

    \begin{tabular}{@{}c@{}}
  
        \prftree[r]{\footnotesize($@$B)}
        {@_ji, \Gamma \vdash \Delta}
        {@_ij, \Gamma \vdash \Delta}

        \qquad
        {\prftree[r]{\footnotesize(@$\cmpr$L)}
        {\tagg{\cmpr}{i}{j},\Gamma \vdash \Delta}
        {@_k\tagg{\cmpr}{i}{j},\Gamma \vdash \Delta}}
  
        \qquad
    
        {\prftree[r]{\footnotesize(@$\cmpr$R)}
        {\Gamma \vdash \Delta,\tagg{\cmpr}{i}{j}}
        {\Gamma \vdash \Delta,@_k\tagg{\cmpr}{i}{j}}}
  
    \end{tabular}
  
    \tabularnewline
  
    \midrule
    \textbf{\textsc{Derived Rules for Data Comparison}} \\
    \midrule
  
    \tabularnewline[-7pt]
    \begin{tabular}{@{}c@{}}
      
      {\prftree[r]{\footnotesize($\tup{\cmpr}$B)}
      {\tagg{\cmpr}{j}{i},\Gamma \vdash \Delta}
      {\tagg{\cmpr}{i}{j},\Gamma \vdash \Delta}}
      

      
  
  
      \tabularnewline[7pt]
  
      {\prftree[r]{\footnotesize($[\cmpr]$L)}
      {@_i[\alpha \cmpr \beta], \tagg{\cmpr}{j}{k}, @_i\tup{\alpha}j, @_i\tup{\beta}k, \Gamma \vdash \Delta}
      {@_i[\alpha \cmpr \beta], @_i\tup{\alpha}j, @_i\tup{\beta}k, \Gamma \vdash \Delta}}
  
      \quad
  
      {\prftree[r]{\footnotesize($[\cmpr]$R)$^{\dagger}$}
      {@_i\tup{\alpha}j, @_i\tup{\beta}k, \Gamma \vdash \Delta, \tagg{\cmpr}{j}{k}}
      {\Gamma \vdash \Delta, @_i[\alpha \cmpr \beta]}}   
      
      \tabularnewline[5pt]
      {\footnotesize $^{\dagger}$ j and k are different and do not occur in the conclusion}
    \end{tabular}
    
    \tabularnewline
    \bottomrule
    \end{tabular}
    }}
    \caption{Derived Sequent Rules for $\hxpd$.}\label[figure]{drules:hxpd}
\end{figure}

%% file: alpha-sigma.tex
\begin{itemize} \item ($\tup{\sigma}$L).
        Let $\sigma \in \{\dowa, \psi?, k{:}\}$.
        The following is a derived rule:

        \begin{center}\scalebox{\thescalefactor}
            {
                \prftree[r]{\footnotesize$(\tup{\sigma}\text{L})^{\dagger}$}
                {
                    @_i\tup{\sigma}j, @_j\varphi, \Gamma \vdash \Delta
                }
                {
                    @_i\tup{\sigma}\varphi, \Gamma \vdash \Delta
                }
            }

            \smallskip
            \scalebox{\thescalefactor}
            {$^{\dagger}$ were $j$ is not in the conclusion}
        \end{center}

        The proof is by a case analysis.

        \begin{itemize}
            \item $\sigma = \dowa$ is ($\tup{\dowa}$L).

            \item $\sigma = k{:}$ is
        
            \begin{center}\scalebox{\thescalefactor}
            {
                {\prftree[r]{\footnotesize(@L)}
                {
                    \prftree[r]{\footnotesize(Nom)}
                    {
                        \prftree[r]{\footnotesize(S$_1$)}
                        {
                            \prftree[r]{\footnotesize(WL)}
                            {
                                \prftree[r]{\footnotesize$(\inv{\text{@L}})$}
                                {
                                    @_i\tup{k{:}}j, @_j\varphi, \Gamma \vdash \Delta
                                }
                                {@_kj, @_j\varphi, \Gamma \vdash \Delta}
                            }
                            {@_j\varphi, @_kj, @_k\varphi, \Gamma \vdash \Delta}
                        }
                        {@_kj, @_k\varphi, \Gamma \vdash \Delta}
                    }
                    {@_k\varphi, \Gamma \vdash \Delta}
                }
                {@_i\tup{k{:}}\varphi, \Gamma \vdash \Delta}}
            }
            \end{center}

            \item $\sigma = \psi?$ is
        
            \begin{center}\scalebox{\thescalefactor}
            {
                \prftree[r]{\footnotesize$({\land}\text{L})$}
                {
                    \prftree[r]{\footnotesize(Nom)}
                    {
                        \prftree[r]{\footnotesize(S$_1$)}
                        {
                            \prftree[r]{\footnotesize(WL)}
                            {
                                \prftree[r]{\footnotesize$(\inv{{\land}\text{L}})$}
                                {
                                    @_i\tup{\psi}j, @_j\varphi, \Gamma \vdash \Delta
                                }
                                {@_j\varphi, @_ij, @_i\psi, \Gamma \vdash \Delta}
                            }
                            {@_j\varphi, @_ij, @_i\psi, @_i\varphi, \Gamma \vdash \Delta}
                        }
                        {@_ij, @_i\psi, @_i\varphi, \Gamma \vdash \Delta}
                    }
                    {@_i\psi, @_i\varphi, \Gamma \vdash \Delta}
                }
                {@_i\tup{\psi?}\varphi, \Gamma \vdash \Delta}
            }
            \end{center}
    \end{itemize}
\end{itemize}

\begin{itemize}
    \item ($\tup{\sigma}$R).
        Let $\sigma \in \{\dowa, \psi?, k{:}\}$.
        The following is a derived rule:

        \begin{center}\scalebox{\thescalefactor}
            {
                \prftree[r]{\footnotesize$(\tup{\sigma}\text{R})$}
                {
                    @_i\tup{\sigma}j, \Gamma \vdash \Delta, @_i\tup{\sigma}\varphi, @_j\varphi
                }
                {
                    @_i\tup{\sigma}j, \Gamma \vdash \Delta, @_i\tup{\sigma}\varphi
                }
            }
        \end{center}

        The proof is by a case analysis.

        \begin{itemize}
            \item $\sigma = \dowa$ is ($\tup{\dowa}$L).

            \item $\sigma = k{:}$ is
        
            \begin{center}\scalebox{\thescalefactor}
            {
                \prftree[r]{\footnotesize(Cut)}
                {
                    \prfassumption
                    {@_i\tup{k{:}}j, \Gamma \vdash \Delta, @_i\tup{k{:}}\varphi, \cc{@_j\varphi}}
                }
                {
                    \prftree[r]{\footnotesize(@L,@R)}
                    {
                        \prftree[r]{\footnotesize(@B)}
                        {
                            \prftree[r]{\footnotesize(S$_1$)}
                            {
                                \prfbyaxiom{\footnotesize(Ax)}
                                {@_k\varphi, @_j\varphi, @_jk, \Gamma \vdash @_k\varphi}
                            }
                            {@_j\varphi, @_jk, \Gamma \vdash @_k\varphi}
                        }
                        {@_j\varphi, @_kj, \Gamma \vdash @_k\varphi}
                    }
                    {\cc{@_j\varphi}, @_i\tup{k{:}}j, \Gamma \vdash \Delta, @_i\tup{k{:}}\varphi}
                }
                {
                    @_i\tup{k{:}}j, \Gamma \vdash \Delta, @_i\tup{k{:}}\varphi
                }
            }
            \end{center}

            \item $\sigma = \psi?$ is
        
            \begin{center}\scalebox{\thescalefactor}
                {
                    \prftree[r,l]{\footnotesize(${\land}\text{L}$,${\land}\text{R}$)}{$\aderivation =$}
                    {
                        \prfbyaxiom{\footnotesize(Ax)}
                        {@_j\varphi, @_i\psi, @_ij, \Gamma \vdash \Delta, @_i\psi}
                    }
                    {
                        \prftree[r]{\footnotesize(S$_1$)}
                        {
                            \prfbyaxiom{\footnotesize(Ax)}
                            {@_i\varphi, @_j\varphi, @_i\psi, @_ij, \Gamma \vdash \Delta, @_i\varphi}
                        }
                        {@_j\varphi, @_i\psi, @_ij, \Gamma \vdash \Delta, @_i\varphi}
                    }
                    {\cc{@_j\varphi}, @_i\tup{\psi?}j, \Gamma \vdash \Delta, @_i\tup{\psi?}\varphi}
                }
            \end{center}

            \vspace{0.5cm}
            \begin{center}\scalebox{\thescalefactor}
            {
                \prftree[r]{\footnotesize(Cut)}
                {
                    \prfassumption
                    {@_i\tup{\psi?}j, \Gamma \vdash \Delta, @_i\tup{\psi?}\varphi, \cc{@_j\varphi}}
                }
                {
                   \prfassumption{\aderivation}
                }
                {
                    @_i\tup{\psi?}j, \Gamma \vdash \Delta, @_i\tup{\psi?}\varphi
                }
            }
            \end{center}
    \end{itemize}
\end{itemize}

\begin{itemize}\item ($\tup{\alpha}$L).
    The proof is by induction the definition of $\tup{\alpha}\varphi$.

    \begin{itemize}
        \item Base Case: $\alpha = \sigma$ is ($\tup{\sigma}$L).

        \item Inductive Step: {$\alpha=\dowa\alpha$}
        
            \begin{center}\scalebox{\thescalefactor}
            {
                {\prftree[r]{\footnotesize($\tup{\sigma}$L)}
                {
                    \prftree[r]{\footnotesize(IH)}
                    {
                        \prftree[r]{\footnotesize(Cut)}
                        {
                            \prftree[r]{\footnotesize($\tup{\sigma}$R)}
                            {
                                \prfbyaxiom{\footnotesize(Ax)}
                                {@_i\tup{\sigma}j, @_j\tup{\alpha}a,\Gamma \vdash \Delta,\cc{@_i\tup{\sigma\alpha}a}, @_k\tup{\alpha}a}
                            }
                            {@_i\tup{\sigma}j, @_j\tup{\alpha}a,\Gamma \vdash \Delta,\cc{@_i\tup{\sigma\alpha}a}}
                        }
                        {
                            \prfassumption
                            {\cc{@_i\tup{\sigma\alpha}a}, @_a\varphi, \Gamma \vdash \Delta}
                        }
                        {@_i\tup{\sigma}j, @_j\tup{\alpha}a, @_a\varphi, \Gamma \vdash \Delta}
                    }
                    {@_i\tup{\sigma}j, \rr{@_j\tup{\alpha}\varphi}, \Gamma \vdash \Delta}
                }
                {@_i\tup{\sigma\alpha}\psi, \Gamma \vdash \Delta}}
            }
            \end{center}
    \end{itemize}
\end{itemize}

\begin{itemize}\item ($\tup{\alpha}$R).
    The proof is by induction the definition of $\tup{\alpha}\varphi$.

    \begin{itemize}

        \item Base Case: $\alpha = \sigma$ is ($\tup{\sigma}$R).

        \item Inductive Step: {$\alpha=\sigma\alpha$}

            \begin{center}\scalebox{\thescalefactor}
                {
                    \prftree[r,l]{\footnotesize($\tup{\sigma}$L)}{$\aderivation =$}
                    {
                        \prftree[r]{\footnotesize($\tup{\sigma}$R)}
                        {
                            \prftree[r]{\footnotesize(IH)}
                            {
                                \prfbyaxiom{\footnotesize(Ax)}
                                {@_i\tup{\sigma}j, @_j\tup{\alpha}a, @_a\varphi, \Gamma \vdash \Delta@_i\tup{\sigma\alpha}\varphi,@_j\tup{\alpha}\varphi,@_j\tup{\alpha}a, @_a\varphi}
                            }
                            {@_i\tup{\sigma}j, \rr{@_j\tup{\alpha}a}, @_a\varphi, \Gamma \vdash \Delta @_i\tup{\sigma\alpha}\varphi,\rr{@_j\tup{\alpha}\varphi}}
                        }
                        {@_i\tup{\sigma}j, @_j\tup{\alpha}a, @_a\varphi, \Gamma \vdash \Delta @_i\tup{\sigma\alpha}\varphi}
                    }
                    {\cc{@_a\varphi},@_i\tup{\sigma\alpha}a, \Gamma \vdash \Delta @_i\tup{\sigma\alpha}\varphi}
                }
            \end{center}
            \vspace{0.5cm}
            \begin{center}\scalebox{\thescalefactor}
            {
                \prftree[r]{\footnotesize(Cut)}
                {@_i\tup{\sigma\alpha}a, \Gamma \vdash \Delta @_i\tup{\sigma\alpha}\varphi, \cc{@_a\varphi}}
                {
                    \prfassumption{\aderivation}
                }
                {@_i\tup{\sigma\alpha}a, \Gamma \vdash \Delta @_i\tup{\sigma\alpha}\varphi}
            }
            \end{center}
    \end{itemize}
\end{itemize}

%% file: prf-ax-s1.tex
\begin{itemize}
    \item (Ax).
    
    \begin{itemize}

        \item Base Case:
            If $\varphi = \bot$, the result is obtained using ($\bot$).
            If $\varphi \in \{p,i\}$, the result is obtained using the basic (Ax).

        \item Inductive Step:
        
        \begin{itemize}
            \item {$\varphi = \varphi \to \psi$}
                \begin{center}\scalebox{\thescalefactor}
                {
                    \prftree[r]{\footnotesize($\to$R)}
                    {
                        \prftree[r]{\footnotesize($\to$L)}
                        {
                            \prfbyaxiom{\footnotesize(IH)}
                            {@_i\varphi, \Gamma \vdash \Delta, @_i\psi, @_i\varphi}
                        }
                        {
                            \prfbyaxiom{\footnotesize(IH)}
                            {@_i\varphi, @_i\psi, \Gamma \vdash \Delta, @_i\psi}
                        }
                        {@_i\varphi, @_i(\varphi \to \psi), \Gamma \vdash \Delta, @_i\psi}
                    }
                    {@_i(\varphi \to \psi), \Gamma \vdash \Delta, @_i(\varphi \to \psi)}
                }
                \end{center}
            \item {$\varphi = \tup{\sigma}\varphi$}
                \begin{center}\scalebox{\thescalefactor}
                {
                    \prftree[r]{\footnotesize($\tup{\sigma}$L)}
                    {
                        \prftree[r]{\footnotesize($\tup{\sigma}$R)}
                        {
                            \prfbyaxiom{\footnotesize(IH)}
                            {@_i\tup{\sigma}j, @_j\varphi, \Gamma \vdash \Delta,@_i\tup{\sigma}\varphi, @_j\varphi}
                        }
                        {@_i\tup{\sigma}j, @_j\varphi, \Gamma \vdash \Delta,@_i\tup{\sigma}\varphi}
                    }
                    {@_i\tup{\sigma}\varphi, \Gamma \vdash \Delta, @_i\tup{\sigma}\varphi}
                }
                \end{center}
            \item {$\varphi = @_j\varphi$}
                \begin{center}\scalebox{\thescalefactor}
                {
                    \prftree[r]{\footnotesize($@$L, $@$R)}
                    {
                        \prfbyaxiom{\footnotesize(IH)}
                        {@_j\varphi, \Gamma \vdash \Delta, @_j\varphi}
                    }
                    {@_i@_j\varphi, \Gamma \vdash \Delta, @_i@_j\varphi}
                }
                \end{center}
            \item {$\varphi = \tup{\alpha \cmpr \beta}$}
                \begin{center}\scalebox{\thescalefactor}
                {
                    \prftree[r]{\footnotesize($\tup{\cmpr}$L)}
                    {
                        \prftree[r]{\footnotesize($\tup{\cmpr}$R)}
                        {
                            \prfbyaxiom{\footnotesize(Ax)}
                            {@_i\tup{\alpha}a, @_i\tup{\beta}b,\tagg{\cmpr}{a}{b}, \Gamma \vdash \Delta,@_i\tup{\alpha \cmpr\beta},\tagg{\cmpr}{a}{b}}
                        }
                        {@_i\tup{\alpha}a, @_i\tup{\beta}b,\tagg{\cmpr}{a}{b}, \Gamma \vdash \Delta,@_i\tup{\alpha \cmpr\beta}}
                    }
                    {@_i\tup{\alpha \cmpr\beta}, \Gamma \vdash \Delta,@_i\tup{\alpha \cmpr\beta}}
                }
                \end{center}

        \end{itemize}
    \end{itemize}

    \item (S$_1$)

    \begin{itemize}
        \item Base Case:
            If $\varphi = j$, the result is obtained using (@5).
            If $\varphi \in \{p,\bot\}$, the result is obtained using the basic (S$_1$).
        \item Inductive Step:
        
        \begin{itemize}
            \item {$\varphi = \varphi \to \psi$}

                \begin{center}\scalebox{\thescalefactor}
                    {
                        \prftree[r,l]{\footnotesize($\to$R)}{$\aderivation =$}
                        {
                            \prftree[r]{\footnotesize($\to$L,WL)}
                            {
                                \prftree[r]{\footnotesize(IH)}
                                {
                                    \prfbyaxiom{\footnotesize(Ax)}
                                    {@_j\psi,@_ij,@_i\psi, \Gamma \vdash \Delta,@_j\psi}
                                }
                                {@_ij,@_i\psi, \Gamma \vdash \Delta,@_j\psi}
                            }
                            {
                                \prftree[r]{\footnotesize($@$B)}
                                {
                                    \prftree[r]{\footnotesize(IH)}
                                    {
                                        \prfbyaxiom{\footnotesize(Ax)}
                                        {@_i\varphi,@_j\varphi,@_ji, \Gamma \vdash \Delta, @_i\varphi}
                                    }
                                    {@_j\varphi,@_ji, \Gamma \vdash \Delta, @_i\varphi}
                                }
                                {@_j\varphi,@_ij, \Gamma \vdash \Delta, @_i\varphi}
                            }
                            {@_j\varphi,@_ij, @_i(\varphi \to \psi), \Gamma \vdash \Delta,@_j\psi }
                        }
                        {@_ij, @_i(\varphi \to \psi), \Gamma \vdash \Delta,\cc{@_j(\varphi \to \psi)}}
                    }
                \end{center}
                \vspace{0.5cm}
                \begin{center}\scalebox{\thescalefactor}
                {
                    \prftree[r]{\footnotesize(Cut)}
                    {
                        \prfassumption{\aderivation}
                    }
                    {
                        \prfassumption{\cc{@_j(\varphi \to \psi)},@_ij, @_i(\varphi \to \psi), \Gamma \vdash \Delta}
                    }
                    {@_ij, @_i(\varphi \to \psi), \Gamma \vdash \Delta}
                }
                \end{center}
            \item {$\varphi = \tup{\sigma}\varphi$}
                \begin{center}\scalebox{\thescalefactor}
                    {
                        \prftree[r,l]{\footnotesize($\tup{\sigma}$L)}{$\aderivation =$}
                        {
                            \prftree[r]{\footnotesize(S$_1$)}
                            {
                                \prftree[r]{\footnotesize($\tup{\sigma}$R)}
                                {
                                    \prfbyaxiom{\footnotesize(Ax)}
                                    {\taag{\sigma}{j}{k},\taag{\sigma}{i}{k}, @_k\varphi,@_ij, \Gamma \vdash \Delta,@_j\tup{\sigma}\varphi, @_k\varphi}
                                }
                                {\taag{\sigma}{j}{k},\taag{\sigma}{i}{k}, @_k\varphi,@_ij, \Gamma \vdash \Delta,@_j\tup{\sigma}\varphi}
                            }
                            {\taag{\sigma}{i}{k}, @_k\varphi,@_ij, \Gamma \vdash \Delta,@_j\tup{\sigma}\varphi}
                        }
                        {@_ij, @_i\tup{\sigma}\varphi, \Gamma \vdash \Delta,\cc{@_j\tup{\sigma}\varphi} }
                    }
                \end{center}
                \begin{center}\scalebox{\thescalefactor}
                {
                    \prftree[r]{\footnotesize(Cut)}
                    {
                        \prfassumption{\aderivation}
                    }
                    {
                        \prfassumption
                        {\cc{@_j\tup{\sigma}\varphi},@_ij, @_i\tup{\sigma}\varphi, \Gamma \vdash \Delta}
                    }
                    {@_ij, @_i\tup{\sigma}\varphi, \Gamma \vdash \Delta}
                }
                \end{center}
            \item {$\varphi = @_j\varphi$}
                \begin{center}\scalebox{\thescalefactor}
                {
                    \prftree[r]{\footnotesize(Cut)}
                    {
                        \prftree[r]{\footnotesize($@$L,$@$R)}
                        {
                            \prfbyaxiom{\footnotesize(Ax)}
                            {@_ij, @_k\varphi, \Gamma \vdash \Delta, @_k\varphi}
                        }
                        {@_ij, @_i@_k\varphi, \Gamma \vdash \Delta, \cc{@_j@_k\varphi}}
                    }
                    {
                        \prfassumption
                        {\cc{@_j@_k\varphi},@_ij, @_i@_k\varphi, \Gamma \vdash \Delta}
                    }
                    {@_ij, @_i@_k\varphi, \Gamma \vdash \Delta}
                }
                \end{center}
            \item {$\varphi = \tup{\alpha \cmpr \beta}$}
                \begin{center}\scalebox{\thescalefactor}
                    {
                        \prftree[r,l]{\footnotesize($\tup{\cmpr}$L)}{$\aderivation = $}
                        {   
                            \prftree[r]{\footnotesize(IH)}
                            {
                                \prftree[r]{\footnotesize($\tup{\cmpr}$R)}
                                {
                                    \prfbyaxiom{\footnotesize(Ax)}
                                    {@_j\tup{\alpha}a, @_j\tup{\beta}b,@_i\tup{\alpha}a, @_i\tup{\beta}b, \tagg{\cmpr}{a}{b}, @_ij, \Gamma \vdash \Delta,@_j\tup{\alpha \cmpr \beta},\tagg{\cmpr}{a}{b}}
                                }
                                {@_j\tup{\alpha}a, @_j\tup{\beta}b,@_i\tup{\alpha}a, @_i\tup{\beta}b, \tagg{\cmpr}{a}{b}, @_ij, \Gamma \vdash \Delta,@_j\tup{\alpha \cmpr \beta}}
                            }
                            {@_i\tup{\alpha}a, @_i\tup{\beta}b, \tagg{\cmpr}{a}{b}, @_ij, \Gamma \vdash \Delta,@_j\tup{\alpha \cmpr \beta}}
                        }
                        {@_ij, @_i\tup{\alpha \cmpr \beta}, \Gamma \vdash \Delta, \cc{@_j\tup{\alpha \cmpr \beta}}}
                    }
                \end{center}

                \hspace{0.5cm}

                \begin{center}\scalebox{\thescalefactor}
                {
                    \prftree[r]{\footnotesize(Cut)}
                    {
                        \prfassumption{\aderivation}
                    }
                    {
                        \prfassumption
                        {\cc{@_j\tup{\alpha \cmpr \beta}},@_ij, @_i\tup{\alpha \cmpr \beta}, \Gamma \vdash \Delta}
                    }
                    {@_ij, @_i\tup{\alpha \cmpr \beta}, \Gamma \vdash \Delta}
                }
                \end{center}
        \end{itemize}
    \end{itemize}

    \item (S$_2$) 
    
    \begin{itemize}
        \item Base Case
        
        \begin{itemize}
            \item $\alpha = \dowa$ corresponds to the (S$_2$).
            \item $\alpha = k{:}$
                \begin{center}\scalebox{\thescalefactor}
                {
                    \prftree[r]{\footnotesize($@$L)}
                    {
                        \prftree[r]{\footnotesize($@$B)}
                        {
                            \prftree[r]{\footnotesize($@$5)}
                            {
                                \prftree[r]{\footnotesize($@$B)}
                                {
                                    \prftree[r]{\footnotesize($\inv{{@}\text{L}}$)}
                                    {@_i\tup{a{:}}k, @_jk, @_i\tup{a{:}}j, \Gamma \vdash \Delta}
                                    {@_ak, @_jk, @_aj, \Gamma \vdash \Delta}
                                }
                                {@_ka,@_jk, @_ja, \Gamma \vdash \Delta}
                            }
                            {@_jk, @_ja, \Gamma \vdash \Delta}
                        }
                        {@_jk, @_aj, \Gamma \vdash \Delta}
                    }
                    {@_jk, @_i\tup{a{:}}j, \Gamma \vdash \Delta}
                }
                \end{center}
            \item $\alpha = \varphi?$
                \begin{center}\scalebox{\thescalefactor}
                {
                    \prftree[r]{\footnotesize($\land$L)}
                    {
                        \prftree[r]{\footnotesize($@$B)}
                        {
                            \prftree[r]{\footnotesize($@$5)}
                            {
                                \prftree[r]{\footnotesize($@$B)}
                                {
                                    \prftree[r]{\footnotesize($\inv{\land\text{L}}$)}
                                    {@_i\tup{\varphi?}k, @_jk, @_i\tup{\varphi?}j, \Gamma \vdash \Delta}
                                    {@_ik, @_jk, @_i\varphi, @_ij, \Gamma \vdash \Delta}
                                }
                                {@_ki,@_jk, @_i\varphi, @_ji, \Gamma \vdash \Delta }
                            }
                            {@_jk, @_i\varphi, @_ji, \Gamma \vdash \Delta}
                        }
                        {@_jk, @_i\varphi, @_ij, \Gamma \vdash \Delta}
                    }
                    {@_jk, @_i\tup{\varphi?}j, \Gamma \vdash \Delta}
                }
                \end{center}
        \end{itemize}

        \item Inductive Step

        \begin{itemize}
            \item {$\alpha=k{:}\alpha$}
                \begin{center}\scalebox{\thescalefactor}
                {
                    \prftree[r]{\footnotesize($@$L)}
                    {
                        \prftree[r]{\footnotesize(IH)}
                        {
                            \prftree[r]{\footnotesize($\inv{{@}\text{L}}$)}
                            {@_i\tup{a{:}\alpha}k, @_jk, @_i\tup{a{:}\alpha}j, \Gamma \vdash \Delta}
                            {@_a\tup{\alpha}k, @_jk, @_a\tup{\alpha}j, \Gamma \vdash \Delta}
                        }
                        {@_jk, @_a\tup{\alpha}j, \Gamma \vdash \Delta}
                    }
                    {@_jk, @_i\tup{a{:}\alpha}j, \Gamma \vdash \Delta}
                }
                \end{center}
            \item {$\alpha=\varphi?\alpha$}
                \begin{center}\scalebox{\thescalefactor}
                {
                    \prftree[r]{\footnotesize($\land$L)}
                    {
                        \prftree[r]{\footnotesize(IH)}
                        {
                            \prftree[r]{\footnotesize($\inv{{\land}\text{L}}$)}
                            {@_i\tup{\varphi?\alpha}k, @_jk,  @_i\tup{\varphi?\alpha}j, \Gamma \vdash \Delta}
                            {@_i\tup{\alpha}k, @_jk, @_i\varphi, @_i\tup{\alpha}j, \Gamma \vdash \Delta}
                        }
                        {@_jk, @_i\varphi, @_i\tup{\alpha}j, \Gamma \vdash \Delta}
                    }
                    {@_jk, @_i\tup{\varphi?\alpha}j, \Gamma \vdash \Delta}
                }
                \end{center}
            \item {$\alpha=\dowa\alpha$}
                \begin{center}\scalebox{\thescalefactor}
                {
                    {\prftree[r]{\footnotesize($\tup{\dowa}$L)}
                    {
                        \prftree[r]{\footnotesize(IH)}
                        {
                            \prftree[r]{\footnotesize($\inv{\tup{\dowa}\text{L}}$)}
                            {
                                \prfassumption
                                {@_i\tup{\dowa\alpha}k, @_jk, @_i\tup{\dowa\alpha}j, \Gamma \vdash \Delta}
                            }
                            {@_a\tup{\alpha}k, @_jk, @_i\tup{\dowa}a, @_a\tup{\alpha}j, \Gamma \vdash \Delta}
                        }
                        {@_jk, @_i\tup{\dowa}a, @_a\tup{\alpha}j, \Gamma \vdash \Delta}
                    }
                    {@_jk, @_i\tup{\dowa\alpha}j, \Gamma \vdash \Delta}}
                }
                \end{center}
        \end{itemize}
    \end{itemize}
    
    \item (S$_3$)

    \begin{center}\scalebox{\thescalefactor}
    {
        \prftree[r]{\footnotesize(Cut)}
        {
            \prftree[r]{\footnotesize(NEqR)}
            {
                \prftree[r]{\footnotesize(NEqL,@B)}
                {
                    \prftree[r]{\footnotesize(S$_3$)}
                    {
                        \prfbyaxiom{\footnotesize(Ax)}
                        {\tagg{=}{i}{k}, \tagg{=}{j}{k}, @_ji, \Gamma \vdash \Delta, \tagg{=}{i}{k}}
                    }
                    {\tagg{=}{j}{k}, @_ji, \Gamma \vdash \Delta, \tagg{=}{i}{k}}
                }
                {\tagg{=}{j}{k}, @_ij, \tagg{\neq}{i}{k}, \Gamma \vdash \Delta}
            }
            {@_ij, \tagg{\neq}{i}{k}, \Gamma \vdash \Delta, \cc{\tagg{\neq}{j}{k}}}
        }
        {
            \prfassumption
            {\cc{\tagg{\neq}{j}{k}}, @_ij, \tagg{\neq}{i}{k}, \Gamma \vdash \Delta}
        }
        {@_ij, \tagg{\neq}{i}{k}, \Gamma \vdash \Delta}
    }
    \end{center}
\end{itemize}

%% file: arXiv-completeness.tex
\section{Completeness}
\label{sec:arVix-completeness}


\begin{lemma}\label[lemma]{lemma:derivability}
    Let $\vdash_{\hilbert}^n \varphi$ indicate that $\varphi$ is a theorem in $\hilbert$, with a derivation of length $n$.
    In addition, let $i$ be a nominal not in $\varphi$.
    It follows that, for any $n$, $\vdash_{\hilbert}^n \varphi$ implies $\vdash_{\gentzen} @_i\varphi$ (i.e, $\vdash @_i\varphi$ is provable in $\gentzen$). 
\end{lemma}

\begin{proof}
    The proof proceeds by (strong) induction on the length of a derivation in $\hilbert$.
    The base case requires us to show that each axiom $\varphi$ of the Hilbert system has a corresponding provable sequent $\vdash @_i\varphi$. 
    For the inductive step, we must show that $\vdash_{\hilbert}^m \psi$ implies $\vdash_{\gentzen} @_i \psi$. 
    We proceed by cases, depending on whether ${\vdash_{\hilbert}^m} \psi$ is obtained using (MP), (Nec), (Name), or (Paste). 
    \paragraph{Base Case.}

    \input{arXiv-app}

    \paragraph{Inductive Case.} The inductive hypothesis (IH) is: for all $1 \leq n < m$, if $\vdash_{\hilbert}^{n} \varphi$, then $\vdash_{\gentzen} @_i\varphi$, for some nominal $i$ not occurring in~$\varphi$.
    \begin{itemize}
        \item (MP).
            Suppose that $\psi\,{=}\,\chi \to \varphi$ and there are $n,n' < m$ s.t.\ ${\vdash_{\hilbert}^{n}}\, \chi$ and ${\vdash_{\hilbert}^{n'}}\, \chi \,{\to}\, \varphi$.
            From the~IH, ${\vdash_{\gentzen}}\, @_i\chi$ and ${\vdash_{\gentzen}}\, @_i(\chi \,{\to}\, \varphi)$.
            The derived rule of (MP) gives us ${\vdash_{\gentzen}}\, @_i \psi$ as required.
            \begin{center}\scalebox{\thescalefactor}
                {
                    \prftree[r]{\footnotesize(MP)}
                    {
                        \pp{\vdash_{\gentzen} @_i\chi} 
                    }
                    {
                        \pp{\vdash_{\gentzen} @_i(\chi \to \varphi)} 
                    }
                    {\pp{\vdash @_i\varphi}}
                }
            \end{center}
        \item (Nec).
            Suppose that $\psi\,{=}\,[\alpha]\chi$ and there is $n\,{<}\,m$ s.t.\ ${\vdash_{\hilbert}^{n}}\,\chi$.
            From the IH, ${\vdash_{\gentzen}}\, @_j\chi$.
            Using (WL) we obtain $@_i\tup{\alpha}j\, {\vdash_{\gentzen}}\, @_j\chi$.
            Finally, $([\alpha]\text{R})$ establishes $\vdash_{\gentzen} @_i\psi$ as required.
            
            \begin{center}\scalebox{\thescalefactor}
                {
                    \prftree[r]{\footnotesize($[\alpha]$R)}
                    {
                        \prftree[r]{\footnotesize(WL)}
                        {
                            \prfassumption
                            {\pp{\vdash_{\gentzen} @_j\chi}}
                        }
                        {@_i\<\alpha\>j \vdash @_j\chi}
                    }
                    {\pp{\vdash @_i[\alpha]\chi}}
                }
            \end{center}
            \item (Name). 
            Suppose that $\psi\,{=}\varphi$ and there is $n\,{<}\,m$ s.t.\ ${\vdash_{\hilbert}^{n}}\,@_i\psi$.
                From the IH, ${\vdash_{\gentzen}}\,@_j@_i\psi$.
                From ($\inv{@\text{R}}$), ${\vdash_{\gentzen}}\,@_i\psi$. 
                \begin{center}\scalebox{\thescalefactor}
                    {
                        \prftree[r]{\footnotesize($\inv{@\text{R}}$)}
                        {
                            \pp{\vdash_{\gentzen} @_j@_i\psi}
                        }
                        {\pp{\vdash @_i\psi}}
                    }
                \end{center}
    
            \item (Paste). 
            Suppose that
                $\psi = \<j{:}\dowa\alpha \cmpr \beta\> \to \chi$, and
                there is $n < m$ such that
                    $\vdash_{\hilbert}^{n} (@_j\<\dowa\>k \land \<k{:}\alpha \cmpr \beta\>) \to \chi$.
            By the IH, it follows that $\vdash_{\gentzen} @_i((@_j\<\dowa\>k \land \<k{:}\alpha \cmpr \beta\>) \to \chi)$.
            The following derivation establishes $\vdash_{\gentzen} @_i\psi$, as required.
            \vspace{-.5cm}
            \begin{center}\scalebox{\thescalefactor}
                {
                    \hspace{-1cm}
                    \prftree[r]{\footnotesize($\to$R)}
                    {
                        \prftree[r]{\footnotesize($\tup{\cmpr}$L)}
                        {
                        \prftree[r]{\footnotesize($@$L)}
                        {
                            \prftree[r]{\footnotesize($\tup{\dowa}$L)}
                            {
                                \prftree[r]{\footnotesize(Cut)}
                                {
                                    \prftree[r]{\footnotesize($\land$R,W$*$)}
                                    {
                                        \prftree[r]{\footnotesize($@$R)}
                                        {
                                        \prfbyaxiom{\footnotesize(Ax)}
                                        {\taag{\dowa}{j}{k} \vdash @_j\tup{\dowa}k}
                                        }
                                        {\taag{\dowa}{j}{k} \vdash @_i@_j\tup{\dowa}k}
                                    }
                                    {
                                        \prftree[r]{\footnotesize($\inv{@\text{L}}$)}
                                        {
                                            \prftree[r]{\footnotesize(${\tup{\cmpr}\text{R}}$,W$*$)}
                                            {
                                                \prftree[r]{\footnotesize(Ax)}
                                                {\tagg{\cmpr}{a}{b} \vdash \tagg{\cmpr}{a}{b}}
                                            }
                                            {@_i\tup{k{:}\alpha}a, @_i\tup{\beta}b, \tagg{\cmpr}{a}{b} \vdash @_i\tup{k{:}\alpha \cmpr \beta}}
                                        }
                                        {\rr{@_k\tup{\alpha}a}, @_i\tup{\beta}b, \tagg{\cmpr}{a}{b} \vdash @_i\tup{k{:}\alpha \cmpr \beta}}
                                    }
                                    {\taag{\dowa}{j}{k},@_k\tup{\alpha}a, @_i\tup{\beta}b, \tagg{\cmpr}{a}{b} \vdash \cc{@_i(@_j\tup{\dowa}k \land \tup{k{:}\alpha \cmpr \beta})}}
                                }
                                {
                                    \prftree[r]{\footnotesize($\inv{{\to}\text{R}}$)}
                                    {
                                        \prfassumption
                                        {\pp{\vdash_{\gentzen} @_i(@_j\tup{\dowa}k \land \tup{k{:}\alpha \cmpr \beta} \to \chi)}}
                                    }
                                    {\cc{@_i(@_j\tup{\dowa}k \land \tup{k{:}\alpha \cmpr \beta})} \vdash @_i \chi}
                                }
                                {\taag{\dowa}{j}{k}, @_k\tup{\alpha}a, @_i\tup{\beta}b, \tagg{\cmpr}{a}{b} \vdash @_i \chi}
                            }
                            {@_j\tup{\dowa\alpha}a, @_i\tup{\beta}b, \tagg{\cmpr}{a}{b} \vdash @_i \chi}
                        }
                        {@_i\tup{j{:}\dowa\alpha}a, @_i\tup{\beta}b, \tagg{\cmpr}{a}{b} \vdash @_i \chi}
                        }
                        {@_i\tup{j{:}\dowa\alpha \cmpr \beta} \vdash @_i \chi}
                    }
                    {\pp{\vdash_{\gentzen} @_i(\tup{j{:}\dowa\alpha \cmpr \beta} \to \chi)}}
                } 
            \end{center}
            In the derivation above, (W$*$) indicates the simultaneous application of the rules (WL) and (WR) and (AxG) is a generation of the rule (Ax). 
    \end{itemize}
\end{proof}

\begin{theorem}[Completeness]\label{th:completeness:gentzen}\label{th:soundness}
	Every valid sequent is provable.
\end{theorem}
\begin{proof}
    Suppose $\gamma_1,\ldots,\gamma_n \vdash \delta_1,\ldots,\delta_m$ is valid.
    From the completeness result in~\cite{ArecesF21}, we know there is $k$ such that $\vdash^k_{\hilbert} \bigwedge_{1 \le i \le n} \gamma_i \to \bigvee_{1 \le j \le m} \delta_j$.
    From \Cref{lemma:derivability}, we get $\vdash_{\gentzen} @_i(\bigwedge_{1 \le i \le n} \gamma_i \to \bigvee_{1 \le j \le m} \delta_j)$.
    This implies $@_i\gamma_1, \ldots, @_i\gamma_n \vdash_{\gentzen} @_i\delta_1, \ldots,@_i\delta_m$, and so $\gamma_1,\ldots,\gamma_n \vdash_{\gentzen} \delta_1,\ldots,\delta_m$ as required.
\end{proof}

%% file: arXiv-app.tex
\input{arXiv-prfbasic}
\input{arXiv-prfequality}
\input{arXiv-prfpaths}

%% file: arXiv-prfbasic.tex
\begin{itemize}
    \item {\small Axiom ($@$-def).
    We prove $\vdash_{\hilbert} @_i\varphi \liff \tup{i{:}\varphi? =_{\compc} i{:}\varphi?}$ implies
    $\vdash_{\gentzen} @_j(@_i\varphi \liff \tup{i{:}\varphi? =_{\compc} i{:}\varphi?})$.}
    \begin{center}\scalebox{\thescalefactor}
        {
            \hspace{-.5cm}
            \prftree[r]{\footnotesize($\liff$)R}
            {
                \prftree[r]{\footnotesize($@$L)}
                {
                    \prftree[r]{\footnotesize($@$T)}
                    {
                        \prftree[r]{\footnotesize($\inv{{\land}\text{L}}$)}
                        {
                            \prftree[r]{\footnotesize($\inv{{@}\text{L}}$)}
                            {
                                \prftree[r]{\footnotesize($\tup{\cmpr}$R)}
                                {
                                    \prftree[r]{\footnotesize(EqT)}
                                    {
                                        \prfbyaxiom{\footnotesize(Ax)}
                                        {\tagg{=_{\compc}}{i}{i},@_j\tup{i{:}\varphi?}k, @_j\tup{i{:}\varphi?}k \vdash @_j\tup{i{:}\varphi? =_{\compc} i{:}\varphi?}, \tagg{=_{\compc}}{i}{i}}
                                    }
                                    {@_j\tup{i{:}\varphi?}i, @_j\tup{i{:}\varphi?}i \vdash @_j\tup{i{:}\varphi? =_{\compc} i{:}\varphi?}, \tagg{=_{\compc}}{i}{i}}
                                }
                                {@_j\tup{i{:}\varphi?}i, @_j\tup{i{:}\varphi?}i \vdash @_j\tup{i{:}\varphi? =_{\compc} i{:}\varphi?}}
                            }
                            {@_i\tup{\varphi?i} \vdash @_j\tup{i{:}\varphi? =_{\compc} i{:}\varphi?}}
                    }
                    {@_ii, @_i\varphi \vdash @_j\tup{i{:}\varphi? =_{\compc} i{:}\varphi?}}
                    }
                    {@_i\varphi \vdash @_j\tup{i{:}\varphi? =_{\compc} i{:}\varphi?}}
                }
                {@_j@_i\varphi \vdash @_j\tup{i{:}\varphi? =_{\compc} i{:}\varphi?}}
            }
            {
                \prftree[r]{\footnotesize($@$R)}
                {
                    \prftree[r]{\footnotesize($\tup{\cmpr}$L)}
                    {
                       \prftree[r]{\footnotesize($@$L)}
                       {
                        \prftree[r]{\footnotesize($\land$L)}
                        {
                            \prfbyaxiom{\footnotesize(Ax)}
                            {@_i\varphi, @_ik,@_i\tup{\varphi?}l,\tagg{=_{\compc}}{k}{l}, \Gamma \vdash \Delta, @_i\varphi}
                        }
                        {\rr{@_i\tup{\varphi?}k}, @_i\tup{\varphi?}l,\tagg{=_{\compc}}{k}{l}, \Gamma \vdash \Delta, @_i\varphi}
                       }
                       {\rr{@_j\tup{i{:}\varphi?}k},@_j\tup{i{:}\varphi?}l,\tagg{=_{\compc}}{k}{l}, \Gamma \vdash \Delta, @_i\varphi}
                    }
                    {@_j\tup{i{:}\varphi? =_{\compc} i{:}\varphi?} \vdash @_i\varphi}
                }
                {@_j\tup{i{:}\varphi? =_{\compc} i{:}\varphi?} \vdash @_j@_i\varphi}
            }
            {\vdash @_j(@_i\varphi \liff \tup{i{:}\varphi? =_{\compc} i{:}\varphi?})}
        }
    \end{center}
    \vspace{0.5cm}
    \item {\small Axiom ($\tup{\alpha}$-def). 
    We prove $\vdash_{\hilbert}\tup{\alpha}\varphi \liff \tup{\alpha\varphi? =_{\compc}\alpha\varphi?}$ implies
    $\vdash_{\gentzen}@_i(\tup{\alpha}\varphi \liff \tup{\alpha\varphi? =_{\compc}\alpha\varphi?})$.}
    \begin{center}\scalebox{\thescalefactor}
        {
            \hspace{-1.5cm}
            \prftree[r]{\footnotesize($\liff$)}
                {   
                    \prftree[r]{\footnotesize($\tup{\alpha}$L)}
                    {
                        \prftree[r]{\footnotesize($@$T)}
                        {
                            \prftree[r]{\footnotesize($\inv{{\land}\text{L}}$)}
                            {
                                \prftree[r]{\footnotesize($\inv{{\tup{\alpha}}\text{L}}$)}
                                {
                                    \prftree[r]{\footnotesize($\tup{\cmpr}$R)}
                                    {
                                        \prftree[r]{\footnotesize(EqT)}
                                        {
                                            \prfbyaxiom{\footnotesize(Ax)}
                                            {\tagg{=_{\compc}}{j}{j}, @_i\tup{\alpha\varphi?}j,@_i\tup{\alpha\varphi?}j,\vdash @_i\tup{\alpha\varphi? =_{\compc}\alpha\varphi?}, \tagg{=_{\compc}}{j}{j}}
                                        }
                                        {@_i\tup{\alpha\varphi?}j,@_i\tup{\alpha\varphi?}j,\vdash @_i\tup{\alpha\varphi? =_{\compc}\alpha\varphi?}, \tagg{=_{\compc}}{j}{j}}
                                    }
                                    {@_i\tup{\alpha\varphi?}j,@_i\tup{\alpha\varphi?}j,\vdash \rr{@_i\tup{\alpha\varphi? =_{\compc}\alpha\varphi?}}}
                                }
                                {@_j\tup{\varphi?}j, @_i\tup{\alpha}j \vdash @_i\tup{\alpha\varphi? =_{\compc}\alpha\varphi?}}
                            }
                            {\rr{@_jj},@_i\tup{\alpha}j, \rr{@_j\varphi} \vdash @_i\tup{\alpha\varphi? =_{\compc}\alpha\varphi?}}
                        }
                        {@_i\tup{\alpha}j, @_j\varphi \vdash @_i\tup{\alpha\varphi? =_{\compc}\alpha\varphi?}}
                    }
                    {@_i\tup{\alpha}\varphi \vdash @_i\tup{\alpha\varphi? =_{\compc}\alpha\varphi?}}
                }
                {
                    \prftree[r]{\footnotesize($\tup{\cmpr}$L)}
                    {
                        \prftree[r]{\footnotesize($\tup{\alpha}$L)}
                        {
                            \prftree[r]{\footnotesize($\land$L)}
                            {
                                \prftree[r]{\footnotesize($\inv{{\tup{\alpha}}\text{L}}$)}
                                {
                                    \prfbyaxiom{\footnotesize(Ax)}
                                    {@_i\tup{\alpha}\varphi,@_kj,@_i\tup{\alpha\varphi?}k,\tagg{=_{\compc}}{j}{k} \vdash @_i\tup{\alpha}\varphi }
                                }
                                {\rr{@_k\varphi}, @_kj, \rr{@_i\tup{\alpha}k},@_i\tup{\alpha\varphi?}k,\tagg{=_{\compc}}{j}{k} \vdash @_i\tup{\alpha}\varphi}
                            }
                            {@_i\tup{\alpha}k, \rr{@_k\tup{\varphi?}j},@_i\tup{\alpha\varphi?}k,\tagg{=_{\compc}}{j}{k} \vdash @_i\tup{\alpha}\varphi}
                        }
                        {\rr{@_i\tup{\alpha\varphi?}j},@_i\tup{\alpha\varphi?}k,\tagg{=_{\compc}}{j}{k} \vdash @_i\tup{\alpha}\varphi}
                    }
                    {@_i\tup{\alpha\varphi? =_{\compc}\alpha\varphi?} \vdash @_i\tup{\alpha}\varphi}
                }
                {\vdash @_i(\tup{\alpha}\varphi \liff \tup{\alpha\varphi? =_{\compc}\alpha\varphi?})}
        }
    \end{center}
    \vspace{0.5cm}
    \item {\small Axiom (K).
    We prove $\vdash_{\hilbert}[\alpha](\varphi \to \psi) \to ([\alpha]\varphi \to [\alpha]\psi)$ implies
    $\vdash_{\gentzen}@_i([\alpha](\varphi \to \psi) \to ([\alpha]\varphi \to [\alpha]\psi))$.} 
    \begin{center}\scalebox{\thescalefactor}
        {
            \prftree[r]{\footnotesize($\to$R)}
            {
                \prftree[r]{\footnotesize($\to$R)}
                {
                    \prftree[r]{\footnotesize($[\alpha]$R)}
                    {
                        \prftree[r]{\footnotesize($[\alpha]$L,WL)}
                        {
                            \prftree[r]{\footnotesize($[\alpha]$L,WL)}
                            {
                                \prftree[r]{\footnotesize($\to$R)}
                                {
                                    \prfbyaxiom{\footnotesize(Ax)}
                                    {@_i\varphi \vdash @_j\psi, @_j\varphi}
                                }
                                {
                                    \prfbyaxiom{\footnotesize(Ax)}
                                    {@_j\psi \vdash @_j\psi}
                                }
                                {@_j(\varphi \to \psi),@_j\varphi\vdash @_j\psi }
                            }
                            {@_j\varphi,@_i\tup{\alpha}j,@_i[\alpha](\varphi \to \psi) \vdash @_j\psi }
                        }
                        {@_i\tup{\alpha}j,\rr{@_i[\alpha]\varphi},@_i[\alpha](\varphi \to \psi) \vdash @_j\psi}
                    }
                    {@_i[\alpha]\varphi,@_i[\alpha](\varphi \to \psi) \vdash @_i[\alpha]\psi}
                }
                {@_i[\alpha](\varphi \to \psi) \vdash @_i([\alpha]\varphi \to [\alpha]\psi)}
            }
            {\vdash @_i([\alpha](\varphi \to \psi) \to ([\alpha]\varphi \to [\alpha]\psi))}
        } 
    \end{center}
    \vspace{0.5cm}
    \item {\small Axiom ($@$K). 
    We prove $\vdash_{\hilbert}@_i(\varphi \to \psi) \to (@_i\varphi \to @_i\psi)$ implies
    $\vdash_{\gentzen}@_j(@_i(\varphi \to \psi) \to (@_i\varphi \to @_i\psi))$.}  
    \begin{center}\scalebox{\thescalefactor}
        {
            \prftree[r]{\footnotesize($\to$R)}
            {
                \prftree[r]{\footnotesize($@$L)}
                {
                    \prftree[r]{\footnotesize($\to$R)}
                    {
                        \prftree[r]{\footnotesize($@$L,$@$R)}
                        {
                            \prftree[r]{\footnotesize($\to$L,WL,WR)}
                            {
                                \prfbyaxiom{\footnotesize(Ax)}
                                {@_i\varphi \vdash @_i\varphi}
                            }
                            {
                                \prfbyaxiom{\footnotesize(Ax)}
                                {@_i\psi \vdash @_i\psi}
                            }
                            {@_i\varphi, @_i(\varphi \to \psi) \vdash @_i\psi}
                        }
                        {@_j@_i\varphi, @_i(\varphi \to \psi) \vdash @_j@_i\psi}
                    }
                    {@_i(\varphi \to \psi) \vdash @_j(@_i\varphi \to @_i\psi)}
                }
                {@_j@_i(\varphi \to \psi) \vdash @_j(@_i\varphi \to @_i\psi)}
            }
            {\vdash @_j(@_i(\varphi \to \psi) \to (@_i\varphi \to @_i\psi))}
        } 
    \end{center}
    \vspace{0.5cm}
    \item {\small Axiom ($@$-refl). 
    We prove $\vdash_{\hilbert}@_ii$ implies $\vdash_{\gentzen}@_j@_ii$.} 
    \begin{center}\scalebox{\thescalefactor}
        {
            \prftree[r]{\footnotesize($@$R)}
            {
                \prftree[r]{\footnotesize(Ref)}
                {
                    \prfbyaxiom{\footnotesize(Ax)}
                    {@_ii \vdash @_ii}
                }
                {\vdash @_ii}
            }
            {\vdash @_j@_ii}
        }
    \end{center}
    \vspace{0.5cm}
    \item {\small Axiom ($@$-self-dual). 
    We prove $\vdash_{\hilbert}\lnot @_i \varphi \liff @_i \lnot \varphi$ implies 
    $\vdash_{\gentzen} @_j(\lnot@_i \varphi \liff @_i \lnot \varphi)$.} 
    \begin{center}\scalebox{\thescalefactor}
        {
            \prftree[r]{\footnotesize($\liff$R)}
            {
                \prftree[r]{\footnotesize($\lnot$L,$@$R)}
                {
                    \prftree[r]{\footnotesize($\lnot$R,$@$R)}
                    {
                    \prfbyaxiom{\footnotesize(Ax)}{@_i\varphi \vdash @_i\varphi}
                    }
                    {\vdash @_j@_i\varphi, @_i\lnot\varphi}
                }
                {@_j\lnot @_i \varphi \vdash @_j@_i \lnot \varphi}
            }
            {
                \prftree[r]{\footnotesize($\lnot$R,$@$L)}
                {
                    \prftree[r]{\footnotesize($\lnot$R,$@$R)}
                    {
                        \prfbyaxiom{\footnotesize(Ax)}{@_i\varphi \vdash @_i\varphi}
                    }
                    {@_i\lnot\varphi,@_j@_i\varphi \vdash}
                }
                {@_j@_i \lnot \varphi \vdash @_j\lnot @_i \varphi}}
            {\vdash @_j(\lnot @_i \varphi \liff @_i \lnot \varphi)}
        }
    \end{center}
    \vspace{0.5cm}
    \item {\small Axiom (intro-$@$).
    We prove $\vdash_{\hilbert} i \to (\varphi \liff @_i \varphi)$ implies
    $\vdash_{\gentzen} @_j(i \to (\varphi \liff @_i \varphi))$.}
    \begin{center}\scalebox{\thescalefactor}
        {
            \prftree[r]{\footnotesize($\to$R)}
            {
                \prftree[r]{\footnotesize($\liff$R)}
                {
                    \prftree[r]{\footnotesize($@$R)}
                    {
                        \prftree[r]{\footnotesize(S$_1$)}
                        {
                            \prfbyaxiom{\footnotesize(Ax)}
                            {@_i\varphi,@_k\varphi,@_k i \vdash @_i \varphi }
                        }
                        {@_k\varphi,@_k i \vdash @_i \varphi}
                    }
                    {@_k\varphi,@_k i \vdash @_k@_i \varphi}
                }
                {
                    \prftree[r]{\footnotesize($@$L)}
                    {
                        \prftree[r]{\footnotesize($@$B)}
                        {
                            \prftree[r]{\footnotesize(S$_1$)}
                            {
                                \prfbyaxiom{\footnotesize(Ax)}
                                {@_k \varphi,@_i\varphi,@_i k \vdash @_k \varphi }
                            }
                            {@_i\varphi,@_i k \vdash @_k \varphi}
                        }
                        {@_i\varphi,@_k i \vdash @_k \varphi}
                    }
                    {@_k@_i\varphi ,@_k i \vdash @_k \varphi}
                }
                {@_k i \vdash @_k(\varphi \liff @_i \varphi)}
            }
            {\vdash @_k (i \to (\varphi \liff @_i \varphi))}
        }
    \end{center}
\end{itemize}

%% file: arXiv-prfequality.tex
\begin{itemize}
    \item {\small Axiom (equal). 
    We prove $\vdash_{\hilbert}\tup{\epsilon =_{\compc} \epsilon}$ implies
    $\vdash_{\gentzen}@_i(\tup{\epsilon =_{\compc} \epsilon})$.} 
    \begin{center}\scalebox{\thescalefactor}
        {
            \prftree[r]{\footnotesize(@T)}
            {
                \prftree[r]{\footnotesize($\top$L)}
                {
                    \prftree[r]{\footnotesize($\inv{{\land}\text{L}}$)}
                    {
                        \prftree[r]{\footnotesize($\tup{\cmpr}\text{R}$)}
                        {
                            \prftree[r]{\footnotesize(EqT)}
                            {
                                \prfbyaxiom{\footnotesize(Ax)}
                                {\tagg{=}{i}{i}, @_i\tup{\epsilon}i \vdash @_i\tup{\epsilon =_{\compc} \epsilon}, \tagg{=}{i}{i}}
                            }
                            {@_i\tup{\epsilon}i \vdash @_i\tup{\epsilon =_{\compc} \epsilon}, \tagg{=}{i}{i}}
                        }
                        {@_i\tup{\epsilon}i \vdash @_i\tup{\epsilon =_{\compc} \epsilon}}
                    }
                    {@_i\top, @_ii \vdash @_i\tup{\epsilon =_{\compc} \epsilon}}
                }
                {@_ii \vdash @_i\tup{\epsilon =_{\compc} \epsilon}}
            }
            {\vdash @_i\tup{\epsilon =_{\compc} \epsilon}}
        }
    \end{center}
    \vspace{0.5cm}
    \item {\small Axiom ($\cmpr$-comm).
    We prove $\vdash_{\hilbert}\tup{\alpha \cmpr \beta} \liff{\beta \cmpr \alpha}$ implies
    $\vdash_{\gentzen}@_i(\tup{\alpha \cmpr \beta} \liff{\beta \cmpr \alpha})$.} 
    \begin{center}\scalebox{\thescalefactor}
        {
            \prftree[r]{\footnotesize($\liff$R)}
            {
                \prftree[r]{\footnotesize($\tup{\cmpr}$L)}
                {
                    \prftree[r]{\footnotesize($\tup{\cmpr}$R)}
                    {
                        \prftree[r]{\footnotesize($\tup{\cmpr}$B)}
                        {
                            \prfbyaxiom{\footnotesize(Ax)}
                            {@_i\tup{\alpha}j, @_i\tup{\beta}k, \tagg{\cmpr}{k}{j} \vdash @_i\tup{\beta \cmpr \alpha},\tagg{\cmpr}{k}{j}}
                        }
                        {@_i\tup{\alpha}j, @_i\tup{\beta}k, \tagg{\cmpr}{j}{k} \vdash @_i\tup{\beta \cmpr \alpha},\tagg{\cmpr}{k}{j}}
                    }
                    {@_i\tup{\alpha}j, @_i\tup{\beta}k, \tagg{\cmpr}{j}{k} \vdash @_i\tup{\beta \cmpr \alpha}}
                }
                {@_i\tup{\alpha \cmpr \beta} \vdash @_i\tup{\beta \cmpr \alpha}}
            }
            {
                \prftree[r]{\footnotesize($\tup{\cmpr}$L)}
                {
                    \prftree[r]{\footnotesize($\tup{\cmpr}$R)}
                    {
                        \prftree[r]{\footnotesize(@$\cmpr$B)}
                        {
                            \prfbyaxiom{\footnotesize(Ax)}
                            {@_i\tup{\beta}k,@_i\tup{\alpha}j, \tagg{\cmpr}{j}{k} \vdash @_i\tup{\alpha \cmpr \beta},\tagg{\cmpr}{j}{k}}
                        }
                        {@_i\tup{\beta}k,@_i\tup{\alpha}j, \tagg{\cmpr}{k}{j} \vdash @_i\tup{\alpha \cmpr \beta},\tagg{\cmpr}{j}{k}}
                    }
                    {@_i\tup{\beta}k,@_i\tup{\alpha}j, \tagg{\cmpr}{k}{j} \vdash @_i\tup{\alpha \cmpr \beta}}
                }
                {@_i\tup{\beta \cmpr \alpha} \vdash @_i\tup{\alpha \cmpr \beta}}
            }
            {@_i(\tup{\alpha \cmpr \beta} \liff \tup{\beta \cmpr \alpha})}
        }
    \end{center}
    \vspace{0.5cm}
    \item {\small Axiom ($\epsilon$-trans).
    $\vdash_{\hilbert}\tup{\alpha =_{\compc} \epsilon} \land \tup{\epsilon =_{\compc} \beta} \to \tup{\alpha =_{\compc} \beta}$ implies 
    $\vdash_{\gentzen}@_i(\tup{\alpha =_{\compc} \epsilon} \land \tup{\epsilon =_{\compc} \beta} \to \tup{\alpha =_{\compc} \beta})$.} 
    \begin{center}\scalebox{\thescalefactor}
        {
            \prftree[r]{\footnotesize($\to$R,$\land$L)}
            {
                \prftree[r]{\footnotesize($\tup{\cmpr}$L)}
                {
                    \prftree[r]{\footnotesize($\land$L,$\tup{\cmpr}$R)}
                    {
                        \prftree[r]{\footnotesize(@5,WL)}
                        {
                            \prftree[r]{\footnotesize($\text{S}_3$,WL)}
                            {
                                \prftree[r]{\footnotesize($\tup{\cmpr}$B)}
                                {
                                    \prftree[r]{\footnotesize(Eq5)}
                                    {
                                        \prfbyaxiom{\footnotesize(Ax)}
                                        {\tagg{=}{a}{b}, \tagg{=}{c}{a}, \tagg{=}{c}{b} \vdash \tagg{=}{a}{b}}
                                    }
                                    {\tagg{=}{c}{a}, \tagg{=}{c}{b} \vdash \tagg{=}{a}{b}}
                                }
                                {\tagg{=}{a}{c}, \tagg{=}{c}{b} \vdash \tagg{=}{a}{b}}
                            }
                            {@_{d}c, \tagg{=}{a}{c}, \tagg{=}{d}{b} \vdash \tagg{=}{a}{b}}
                            }
                        {@_ic, @_id, @_i\top, \tagg{=}{a}{c}, @_i\tup{\alpha}a, @_i\tup{\beta}b, \tagg{=}{d}{b} \vdash @_i\tup{\alpha =_{\compc} \beta},\tagg{=}{a}{b}}
                    }
                    {@_i\tup{\alpha}a, \rr{@_i\tup{\epsilon}c}, \tagg{=}{a}{c}, \rr{@_i\tup{\epsilon}d}, @_i\tup{\beta}b, \tagg{=}{d}{b} \vdash @_i\tup{\alpha =_{\compc} \beta}}
                }
                {@_i\tup{\alpha =_{\compc} \epsilon}, @_i\tup{\epsilon =_{\compc} \beta} \vdash @_i\tup{\alpha =_{\compc} \beta}}
            }
            {\vdash @_i(\tup{\alpha =_{\compc} \epsilon} \land \tup{\epsilon =_{\compc} \beta} \to \tup{\alpha =_{\compc} \beta})}
        }
    \end{center}
    \vspace{0.5cm}
    \item {\small Axiom (distinct). 
    We prove $\vdash_{\hilbert}\lnot \tup{\epsilon \neq_{\compc} \epsilon}$ implies
    $\vdash_{\gentzen}@_i(\lnot \tup{\epsilon \neq_{\compc} \epsilon})$.}
    \begin{center}\scalebox{\thescalefactor}
        {
            \prftree[r]{\footnotesize($\lnot$R)}
            {
                \prftree[r]{\footnotesize($\tup{\cmpr}$L)}
                {
                    \prftree[r]{\footnotesize($\land$L)}
                    {
                        \prftree[r]{\footnotesize(@B,WL)}
                        {
                            \prftree[r]{\footnotesize($\text{S}_3$, WL)}
                            {
                                \prftree[r]{\footnotesize($\tup{\cmpr}$B,WL)}
                                {
                                    \prftree[r]{\footnotesize($\text{S}_3$,WL)}
                                    {
                                        \prftree[r]{\footnotesize(NEqL)}
                                        {
                                            \prftree[r]{\footnotesize(EqT)}
                                            {
                                                \prfbyaxiom{\footnotesize(Ax)}
                                                {\tagg{=}{i}{i} \vdash \tagg{=}{i}{i}}
                                            }
                                            {\vdash \tagg{=}{i}{i}}
                                        }
                                        {\tagg{\neq}{i}{i} \vdash}
                                    }
                                    {\tagg{\neq}{k}{i}, @_ki \vdash}
                                }
                                {\rr{\tagg{\neq}{i}{k}}, @_ki \vdash}
                            }
                            {@_ji, @_ki, \tagg{\neq}{j}{k} \vdash}
                        }
                        {@_i\top, \rr{@_ij}, \rr{@_ik}, \tagg{\neq}{j}{k} \vdash}
                    }
                    {\rr{@_i\<\epsilon\>j}, \rr{@_i\<\epsilon\>k}, \tagg{\neq}{j}{k} \vdash}
                }
                {@_i\tup{\epsilon \neq_{\compc} \epsilon} \vdash}
            }
            {\vdash @_i \lnot \tup{\epsilon \neq_{\compc} \epsilon}}
        }
    \end{center}
    \vspace{0.5cm}
    \item {\small Axiom ($@$-data). 
    We prove $\vdash_{\hilbert}\lnot\tup{i{:} = j{:}} \liff \tup{i{:} \neq j{:}}$ implies
    $\vdash_{\gentzen}@_k(\lnot\tup{i{:} = j{:}} \liff \tup{i{:} \neq j{:}})$.}
    \begin{center}\scalebox{\thescalefactor}
        {
            \prftree[r]{\footnotesize($\liff$L)}
            { 
                \prftree[r]{\footnotesize($\lnot$L)}
                { 
                    \prftree[r]{\footnotesize(@$\cmpr$R)}
                    {
                        \prftree[r]{\footnotesize(NEqR)}
                        {
                            \prfbyaxiom{\footnotesize(Ax)}
                            {\tup{i{:} = j{:}} \vdash \tup{i{:} = j{:}}}
                        }
                        {\vdash \tup{i{:} \neq j{:}}, \tup{i{:} = j{:}}}
                    }
                    {\vdash @_k\tup{i{:} \neq j{:}}, @_k\tup{i{:} = j{:}}}
                }
                {@_k\lnot\tup{i{:} = j{:}} \vdash @_k\tup{i{:} \neq j{:}}}
            }
            {
                \prftree[r]{\footnotesize($\lnot$R)}
                { 
                    \prftree[r]{\footnotesize(@$\cmpr$L)}
                    {
                        \prftree[r]{\footnotesize(NEqL)}
                        {
                            \prfbyaxiom{\footnotesize(Ax)}
                            {\tup{i{:} = j{:}} \vdash \tup{i{:} = j{:}}}
                        }
                        {\tup{i{:} = j{:}},\tup{i{:} \neq j{:}} \vdash}
                    }
                    {@_k\tup{i{:} = j{:}}, @_k\tup{i{:} \neq j{:}} \vdash}
                }
                {@_k\tup{i{:} \neq j{:}} \vdash @_k\lnot\tup{i{:} = j{:}}}
            }
            {\vdash @_k(\lnot\tup{i{:} = j{:}} \liff \tup{i{:} \neq j{:}})}
        }
    \end{center}
    \vspace{0.5cm}
    \item {\small Axiom (Subpath).
    We prove $\vdash_{\hilbert}\tup{\alpha \cmpr \beta} \to \tup{\alpha}\top$ implies
    $\vdash_{\gentzen}@_i(\tup{\alpha \cmpr \beta} \to \tup{\alpha}\top)$.} 
    \begin{center}\scalebox{\thescalefactor}
        {
            \prftree[r]{\footnotesize($\to$R)}
            {
                \prftree[r]{\footnotesize($\tup{\cmpr}$L)}
                {
                    \prftree[r]{\footnotesize($\tup{\alpha}$R)}
                    {
                        \prftree[r]{\footnotesize($\top$L)}
                        {
                            \prfbyaxiom{\footnotesize(Ax)}
                            {@_j\top,@_i\tup{\alpha}j, @_i\tup{\beta}k, \tagg{\cmpr}{j}{k} \vdash @_i\tup{\alpha}\top, @_j\top}
                        }
                        {@_i\tup{\alpha}j, @_i\tup{\beta}k, \tagg{\cmpr}{j}{k} \vdash @_i\tup{\alpha}\top, @_j\top}
                    }
                    {@_i\tup{\alpha}j, @_i\tup{\beta}k, \tagg{\cmpr}{j}{k} \vdash @_i\tup{\alpha}\top}
                }
                {@_i\tup{\alpha \cmpr \beta} \vdash @_i\tup{\alpha}\top}
            }
            {\vdash @_i(\tup{\alpha \cmpr \beta} \to \tup{\alpha}\top)}
        }
    \end{center}
    \vspace{0.5cm}
    \item {\small Axiom ($@\cmpr$-dist).
    We prove $\vdash_{\hilbert}\tup{i{:}\alpha \cmpr i{:}\beta} \liff @_i\tup{\alpha \cmpr \beta}$ implies
    $\vdash_{\gentzen}@_j(\tup{i{:}\alpha \cmpr i{:}\beta} \liff @_i\tup{\alpha \cmpr \beta})$.}
    \begin{center}\scalebox{\thescalefactor}
        {
            \hspace{-1.5cm}
            \prftree[r]{\footnotesize($\liff$R)}
            {
                \prftree[r]{\footnotesize($@$R)}
                {
                    \prftree[r]{\footnotesize($\tup{\cmpr}$L)}
                    {
                        \prftree[r]{\footnotesize($@$L)}
                        {
                            \prftree[r]{\footnotesize($\tup{\cmpr}$R)}
                            {
                                \prfbyaxiom{\footnotesize(Ax)}
                                {@_i\tup{\alpha}a, @_i\tup{\beta}b, \tagg{\cmpr}{a}{b} \vdash @_i\tup{\alpha \cmpr \beta}, \tagg{\cmpr}{a}{b}}
                            }
                            {@_i\tup{\alpha}a, @_i\tup{\beta}b, \tagg{\cmpr}{a}{b} \vdash @_i\tup{\alpha \cmpr \beta}}
                        }
                        {@_j\tup{i{:}\alpha}a, @_j\tup{i{:}\beta}b, \tagg{\cmpr}{a}{b} \vdash @_i\tup{\alpha \cmpr \beta}}
                    }
                    {@_j\tup{i{:}\alpha \cmpr i{:}\beta} \vdash @_i\tup{\alpha \cmpr \beta}}
                }
                {@_j\tup{i{:}\alpha \cmpr i{:}\beta} \vdash @_j@_i\tup{\alpha \cmpr \beta}}
            }
            {
                \prftree[r]{\footnotesize($@$L)}
                {
                    \prftree[r]{\footnotesize($\tup{\cmpr}$L)}
                    {
                        \prftree[r]{\footnotesize($\inv{{@}\text{L}}$)}
                        {
                            \prftree[r]{\footnotesize($\tup{\cmpr}$R)}
                            {
                                \prfbyaxiom{\footnotesize(Ax)}
                                {@_j\tup{i{:}\alpha}a, @_j\tup{i{:}\beta}b, \tagg{\cmpr}{a}{b} \vdash @_j\tup{i{:}\alpha \cmpr i{:}\beta}, \tagg{\cmpr}{a}{b}}
                            }
                            {@_j\tup{i{:}\alpha}a, @_j\tup{i{:}\beta}b, \tagg{\cmpr}{a}{b} \vdash @_j\tup{i{:}\alpha \cmpr i{:}\beta}}
                        }
                        {@_i\tup{\alpha}a,@_i\tup{\beta}b, \tagg{\cmpr}{a}{b} \vdash @_j\tup{i{:}\alpha \cmpr i{:}\beta}}
                    }
                    {@_i\tup{\alpha \cmpr \beta} \vdash @_j\tup{i{:}\alpha \cmpr i{:}\beta}}
                }
                {@_j@_i\tup{\alpha \cmpr \beta} \vdash @_j\tup{i{:}\alpha \cmpr i{:}\beta}}
            }
            {\vdash @_j(\tup{i{:}\alpha \cmpr i{:}\beta} \liff @_i\tup{\alpha \cmpr \beta})}
        }        
    \end{center}
    \vspace{0.5cm}
    \item {\small Axiom ($\cmpr$-test).
    We prove $\vdash_{\hilbert}\tup{\varphi?\alpha \cmpr \beta} \liff (\varphi \land \tup{\alpha \cmpr \beta})$ implies 
    $\vdash_{\gentzen}@_i(\tup{\varphi?\alpha \cmpr \beta} \liff (\varphi \land \tup{\alpha \cmpr \beta}))$.}
    \begin{center}\scalebox{\thescalefactor}
        {
            \prftree[r,l]{\footnotesize($\tup{\cmpr}$L)}{$\aderivation=$}
            {
                \prftree[r]{\footnotesize($\land$L)}
                {
                    \prftree[r]{\footnotesize($\land$R)}
                    {
                        \prfbyaxiom{\footnotesize(Ax)}
                        {@_i\varphi, @_i\tup{\alpha}j,@_i\tup{\beta}k, \tagg{\cmpr}{j}{k} \vdash @_i\varphi}
                    }
                    {
                        \prftree[r]{\footnotesize($\tup{\cmpr}$R)}
                        {
                            \prfbyaxiom{\footnotesize(Ax)}
                            {@_i\varphi, @_i\tup{\alpha}j,@_i\tup{\beta}k, \tagg{\cmpr}{j}{k} \vdash @_i\tup{\alpha \cmpr \beta}, \tagg{\cmpr}{j}{k}}
                        }
                        {@_i\varphi, @_i\tup{\alpha}j,@_i\tup{\beta}k, \tagg{\cmpr}{j}{k} \vdash @_i\tup{\alpha \cmpr \beta}}
                    }
                    {@_i\varphi, @_i\tup{\alpha}j,@_i\tup{\beta}k, \tagg{\cmpr}{j}{k} \vdash @_i(\varphi \land \tup{\alpha \cmpr \beta})}
                }
                {@_i\tup{\varphi?\alpha}j, @_i\tup{\beta}k, \tagg{\cmpr}{j}{k} \vdash @_i(\varphi \land \tup{\alpha \cmpr \beta})}
            }
            {@_i\tup{\varphi?\alpha \cmpr \beta} \vdash @_i(\varphi \land \tup{\alpha \cmpr \beta})}
        }
    \end{center}
    \begin{center}\scalebox{\thescalefactor}
        {
            \prftree[r]{\footnotesize($\liff$R)}
            {
                \prfassumption
                {\aderivation}
            }
            {
                \prftree[r]{\footnotesize($\land$L)}
                {
                    \prftree[r]{\footnotesize($\tup{\cmpr}$L)}
                    {
                        \prftree[r]{\footnotesize($\inv{{\land}\text{L}}$)}
                        {
                            \prftree[r]{\footnotesize($\tup{\cmpr}$R)}
                            {
                                \prfbyaxiom{\footnotesize(Ax)}
                                {@_i\tup{\varphi?\alpha}j,@_i\tup{\beta}k, \tagg{\cmpr}{j}{k} \vdash @_i\tup{\varphi?\alpha \cmpr \beta}, \tagg{\cmpr}{j}{k}}
                            }
                            {@_i\tup{\varphi?\alpha}j,@_i\tup{\beta}k, \tagg{\cmpr}{j}{k} \vdash @_i\tup{\varphi?\alpha \cmpr \beta}}
                        }
                        {@_i\tup{\alpha}j, @_i\tup{\beta}k, \tagg{\cmpr}{j}{k},@_i\varphi \vdash @_i\tup{\varphi?\alpha \cmpr \beta}}
                    }
                    {@_i\varphi, @_i\tup{\alpha \cmpr \beta} \vdash @_i\tup{\varphi?\alpha \cmpr \beta}}
                }
                {@_i(\varphi \land \tup{\alpha \cmpr \beta}) \vdash @_i\tup{\varphi?\alpha \cmpr \beta}}
            }
            {\vdash @_i(\tup{\varphi?\alpha \cmpr \beta} \liff (\varphi \land \tup{\alpha \cmpr \beta}))}
        }
    \end{center}
    \vspace{0.5cm}
    \item {\small Axiom (Agree). 
    We prove $\vdash_{\hilbert}\tup{i{:}\alpha \cmpr \beta} \to \tup{j{:}i{:}\alpha \cmpr \beta}$ implies
    $\vdash_{\gentzen}@_k(\tup{i{:}\alpha \cmpr \beta} \to \tup{j{:}i{:}\alpha \cmpr \beta})$.}
    \begin{center}\scalebox{\thescalefactor}
        {
            \prftree[r]{\footnotesize($\to$R)}
            {
                \prftree[r]{\footnotesize($\tup{\cmpr}$L)}
                {
                    \prftree[r]{\footnotesize(@L)}
                    {
                        \prftree[r]{\footnotesize($\inv{@\text{L}}$)}
                        {
                            \prftree[r]{\footnotesize($\inv{@\text{L}}$)}
                            {
                                \prftree[r]{\footnotesize($\inv{\tup{\cmpr}\text{L}}$)}
                                {
                                    \prfbyaxiom{\footnotesize(Ax)}
                                    {@_k\tup{j{:}i{:}\alpha \cmpr \beta} \vdash @_k\tup{j{:}i{:}\alpha \cmpr \beta}}
                                }
                                {@_k\tup{j{:}i{:}\alpha}a,@_k\tup{\beta}b, \tagg{\cmpr}{a}{b} \vdash @_k\tup{j{:}i{:}\alpha \cmpr \beta}}
                            }
                            {\rr{@_j\tup{i{:}\alpha}a},@_k\tup{\beta}b, \tagg{\cmpr}{a}{b} \vdash @_k\tup{j{:}i{:}\alpha \cmpr \beta}}
                        }
                        {\rr{@_i\tup{\alpha}a},@_k\tup{\beta}b, \tagg{\cmpr}{a}{b} \vdash @_k\tup{j{:}i{:}\alpha \cmpr \beta}}
                    }
                    {\rr{@_k\tup{i{:}\alpha}a},@_k\tup{\beta}b, \tagg{\cmpr}{a}{b} \vdash @_k\tup{j{:}i{:}\alpha \cmpr \beta}}
                }
                {@_k\tup{i{:}\alpha \cmpr \beta} \vdash @_k\tup{j{:}i{:}\alpha \cmpr \beta}}
            }
            {\vdash @_k(\tup{i{:}\alpha \cmpr \beta} \to \tup{j{:}i{:}\alpha \cmpr \beta})}
        }
    \end{center}
    \vspace{0.5cm}
    \item {\small Axiom (Back).
    We prove $\vdash_{\hilbert}\tup{\gamma i{:}\alpha \cmpr \beta} \to \tup{i{:}\alpha \cmpr \beta}$ implies
    $\vdash_{\gentzen}@_j(\tup{\gamma i{:}\alpha \cmpr \beta} \to \tup{i{:}\alpha \cmpr \beta})$.}
    \begin{center}\scalebox{\thescalefactor}
        {
            \prftree[r]{\footnotesize($\to$R)}
            {
                \prftree[r]{\footnotesize($\tup{\cmpr}$L)}
                    {
                        \prftree[r]{\footnotesize($\tup{\alpha}$L)}
                        {
                            \prftree[r]{\footnotesize(@L,WL)}
                            {
                                \prftree[r]{\footnotesize($\inv{\text{@L}}$)}
                                {
                                    \prftree[r]{\footnotesize($\inv{\tup{\cmpr}\text{L}}$)}
                                    {
                                        \prfbyaxiom{\footnotesize(Ax)}
                                        {@_k\tup{i{:}\alpha \cmpr \beta} \vdash @_k\tup{i{:}\alpha \cmpr \beta}}
                                    }
                                    {@_k\tup{i{:}\alpha}a, @_k\tup{\beta}b, \tagg{\cmpr}{a}{b} \vdash @_k\tup{i{:}\alpha \cmpr \beta}}
                                }
                                {\rr{@_i\tup{\alpha}a}, @_k\tup{\beta}b, \tagg{\cmpr}{a}{b} \vdash @_k\tup{i{:}\alpha \cmpr \beta}}
                            }
                            {@_k\tup{\gamma}c, \rr{@_c\tup{i{:}\alpha}a}, @_k\tup{\beta}b, \tagg{\cmpr}{a}{b} \vdash @_k\tup{i{:}\alpha \cmpr \beta}}
                        }
                        {\rr{@_k\tup{\gamma i{:}\alpha}a}, @_k\tup{\beta}b, \tagg{\cmpr}{a}{b} \vdash @_k\tup{i{:}\alpha \cmpr \beta}}
                    }
                    {@_k\tup{\gamma i{:}\alpha \cmpr \beta} \vdash @_k\tup{i{:}\alpha \cmpr \beta}}
            }
            {\vdash @_k(\tup{\gamma i{:}\alpha \cmpr \beta} \to \tup{i{:}\alpha \cmpr \beta})}
        }
    \end{center}
    \vspace{0.5cm}
    \item {\small Axiom ($\cmpr$-comp-dist).
    We prove $\vdash_{\hilbert}\tup{\alpha}\tup{\beta \cmpr \gamma} \to \tup{\alpha\beta \cmpr \alpha\gamma}$ implies
    $\vdash_{\gentzen}@_i(\tup{\alpha}\tup{\beta \cmpr \gamma} \to \tup{\alpha\beta \cmpr \alpha\gamma})$.}  
    \begin{center}\scalebox{\thescalefactor}
        {
            \prftree[r]{\footnotesize($\to$R)}
        {
            \prftree[r]{\footnotesize($\tup{\alpha}$L)}
            {
                \prftree[r]{\footnotesize($\tup{\cmpr}$L)}
                {
                    \prftree[r]{\footnotesize($\inv{\tup{\alpha}\text{L}}$)}
                    {
                        \prftree[r]{\footnotesize($\inv{\tup{\cmpr}\text{L}}$)}
                        {
                            \prfbyaxiom{\footnotesize(Ax)}
                            {@_i\tup{\alpha\beta \cmpr \alpha\gamma} \vdash @_i\tup{\alpha\beta \cmpr \alpha\gamma}}
                        }
                        {@_i\tup{\alpha\beta}b, @_i\tup{\alpha\gamma}c, \tagg{\cmpr}{b}{c} \vdash @_i\tup{\alpha\beta \cmpr \alpha\gamma}}
                    }
                    {@_a\tup{\beta}b, @_a\tup{\gamma}c, \tagg{\cmpr}{b}{c}, @_i\tup{\alpha}a \vdash @_i\tup{\alpha\beta \cmpr \alpha\gamma}}
                }
                {@_i\tup{\alpha}a, @_a\tup{\beta \cmpr \gamma} \vdash @_i\tup{\alpha\beta \cmpr \alpha\gamma}}
            }
            {@_i\tup{\alpha}\tup{\beta \cmpr \gamma} \vdash @_i\tup{\alpha\beta \cmpr \alpha\gamma}}
        }
        {\vdash @_i(\tup{\alpha}\tup{\beta \cmpr \gamma} \to \tup{\alpha\beta \cmpr \alpha\gamma})}
        }
    \end{center}
    \vspace{0.5cm}
\end{itemize}

%% file: arXiv-prfpaths.tex
\begin{itemize}
    \item {\small Axiom (comp-assoc).
    $\vdash_{\hilbert}\tup{(\alpha\beta)\gamma \cmpr \eta} \liff \tup{\alpha(\beta\gamma) \cmpr \eta}$ implies
    $\vdash_{\gentzen}@_i(\tup{(\alpha\beta)\gamma \cmpr \eta} \liff \tup{\alpha(\beta\gamma) \cmpr \eta})$.}
    \begin{center}\scalebox{\thescalefactor}
        {
            \hspace{-2cm}
            \prftree[r]{\footnotesize($\liff$R)}
            {
                \prftree[r]{\footnotesize($\tup{\cmpr}$L)}
                {
                    \prftree[r]{\footnotesize($\tup{\alpha}$L)}
                    {
                        \prftree[r]{\footnotesize($\tup{\alpha}$L)}
                        {
                            \prftree[r]{\footnotesize($\inv{{\tup{\alpha}}\text{L}}$)}
                            {
                                \prftree[r]{\footnotesize($\inv{{\tup{\alpha}}\text{L}}$)}
                                {
                                    \prftree[r]{\footnotesize($\tup{\cmpr}$R)}
                                    {
                                        \prfbyaxiom{\footnotesize(Ax)}
                                        {@_i\tup{\alpha(\beta\gamma)}j,@_i\tup{\eta}k, \tagg{\cmpr}{j}{k}, \vdash @_i\tup{\alpha(\beta\gamma) \cmpr \eta}, \tagg{\cmpr}{j}{k}}
                                    }
                                    {@_i\tup{\alpha(\beta\gamma)}j,@_i\tup{\eta}k, \tagg{\cmpr}{j}{k}, \vdash @_i\tup{\alpha(\beta\gamma) \cmpr \eta}}
                                }
                                {@_i\tup{\alpha}a,@_a\tup{\beta\gamma}j,@_i\tup{\eta}k, \tagg{\cmpr}{j}{k}, \vdash @_i\tup{\alpha(\beta\gamma) \cmpr \eta}}
                            }
                            {@_i\tup{\alpha}a,@_a\tup{\beta}b, @_b\tup{\gamma}j,@_i\tup{\eta}k, \tagg{\cmpr}{j}{k}, \vdash @_i\tup{\alpha(\beta\gamma) \cmpr \eta}}
                        }
                        {@_i\tup{\alpha\beta}b, @_b\tup{\gamma}j, @_i\tup{\eta}k, \tagg{\cmpr}{j}{k}, \vdash @_i\tup{\alpha(\beta\gamma) \cmpr \eta}}
                    }
                    {@_i\tup{(\alpha\beta)\gamma}j, @_i\tup{\eta}k, \tagg{\cmpr}{j}{k}, \vdash @_i\tup{\alpha(\beta\gamma) \cmpr \eta}}
                }
                {@_i\tup{(\alpha\beta)\gamma \cmpr \eta} \vdash @_i\tup{\alpha(\beta\gamma) \cmpr \eta}}
            }
            {
            \prftree[r]{\footnotesize($\tup{\cmpr}$L)}
            {
                \prftree[r]{\footnotesize($\tup{\alpha}$L)}
                {
                    \prftree[r]{\footnotesize($\tup{\alpha}$L)}
                    {
                        \prftree[r]{\footnotesize($\inv{{\tup{\alpha}}\text{L}}$)}
                        {
                            \prftree[r]{\footnotesize($\inv{{\tup{\alpha}}\text{L}}$)}
                            {
                                \prftree[r]{\footnotesize($\tup{\cmpr}$R)}
                                {
                                    \prfbyaxiom{\footnotesize(Ax)}
                                    {@_i\tup{(\alpha\beta)\gamma}j, @_i\tup{\eta}k, \tagg{\cmpr}{j}{k},\vdash @_i\tup{(\alpha\beta)\gamma \cmpr \eta}, \tagg{\cmpr}{j}{k}}
                                }
                                {@_i\tup{(\alpha\beta)\gamma}j, @_i\tup{\eta}k, \tagg{\cmpr}{j}{k},\vdash @_i\tup{(\alpha\beta)\gamma \cmpr \eta}}
                            }
                            {@_i\tup{\alpha\beta}b, @_b{\gamma}j, @_i\tup{\eta}k, \tagg{\cmpr}{j}{k},\vdash @_i\tup{(\alpha\beta)\gamma \cmpr \eta}}
                        }
                        {@_i\tup{\alpha}a, @_a{\beta}b, @_b{\gamma}j, @_i\tup{\eta}k, \tagg{\cmpr}{j}{k},\vdash @_i\tup{(\alpha\beta)\gamma \cmpr \eta}}
                    }
                    {@_i\tup{\alpha}a, @_a{(\beta\gamma)}j, @_i\tup{\eta}k, \tagg{\cmpr}{j}{k},\vdash @_i\tup{(\alpha\beta)\gamma \cmpr \eta}}
                }
                {@_i\tup{\alpha(\beta\gamma)}j, @_i\tup{\eta}k, \tagg{\cmpr}{j}{k},\vdash @_i\tup{(\alpha\beta)\gamma \cmpr \eta}}
            }
            {@_i\tup{\alpha(\beta\gamma) \cmpr \eta} \vdash @_i\tup{(\alpha\beta)\gamma \cmpr \eta}}
            }
            {\vdash @_i(\tup{(\alpha\beta)\gamma \cmpr \eta} \liff \tup{\alpha(\beta\gamma) \cmpr \eta})}
        }
    \end{center}
    \vspace{0.5cm}
    \item {\small Axiom (comp-neutral).
    We prove $\vdash_{\hilbert}\tup{\alpha\epsilon\beta \cmpr \gamma} \liff \tup{\alpha\beta \cmpr \gamma}$ implies
    $\vdash_{\gentzen}@_i(\tup{\alpha\epsilon\beta \cmpr \gamma} \liff \tup{\alpha\beta \cmpr \gamma})$.}
    \begin{center}\scalebox{\thescalefactor}
        {
            \hspace{-2.5cm}
            \prftree[r]{\footnotesize($\liff$R)}
            {
                \prftree[r]{\footnotesize($\tup{\cmpr}$L)}
                {
                    \prftree[r]{\footnotesize($\tup{\alpha}$L)}
                    {
                        \prftree[r]{\footnotesize($\land$L)}
                        {
                            \prftree[r]{\footnotesize($\inv{{\tup{\alpha}\text{L}}}$,WL)}
                            {
                                \prftree[r]{\footnotesize($\tup{\cmpr}$R)}
                                {
                                    \prfbyaxiom{\footnotesize(Ax)}
                                    {@_i\tup{\alpha\beta}j, @_i\tup{\gamma}k, \tagg{\cmpr}{j}{k} \vdash @_i\tup{\alpha\beta \cmpr \gamma}, \tagg{\cmpr}{j}{k}}
                                }
                                {@_i\tup{\alpha\beta}j, @_i\tup{\gamma}k, \tagg{\cmpr}{j}{k} \vdash @_i\tup{\alpha\beta \cmpr \gamma}}
                            }
                            {@_a\top, @_a\tup{\beta}j,@_i\tup{\alpha}a, @_i\tup{\gamma}k, \tagg{\cmpr}{j}{k} \vdash @_i\tup{\alpha\beta \cmpr \gamma}}
                        }
                        {@_i\tup{\alpha}a,@_a\tup{\epsilon\beta}j, @_i\tup{\gamma}k, \tagg{\cmpr}{j}{k} \vdash @_i\tup{\alpha\beta \cmpr \gamma}}
                    }
                    {@_i\tup{\alpha\epsilon\beta}j, @_i\tup{\gamma}k, \tagg{\cmpr}{j}{k} \vdash @_i\tup{\alpha\beta \cmpr \gamma}}
                }
                {@_i\tup{\alpha\epsilon\beta \cmpr \gamma} \vdash @_i\tup{\alpha\beta \cmpr \gamma}}
            }
            {
                \prftree[r]{\footnotesize($\tup{\cmpr}$L)}
                {
                    \prftree[r]{\footnotesize($\tup{\alpha}$L)}
                    {
                        \prftree[r]{\footnotesize($\top$L, $@$T)}
                        {
                            \prftree[r]{\footnotesize($\inv{{\land}\text{L}}$)}
                            {
                                \prftree[r]{\footnotesize($\inv{{\tup{\alpha}}\text{L}}$)}
                                {
                                    \prftree[r]{\footnotesize($\inv{{\tup{\alpha}}\text{L}}$)}
                                    {
                                        \prftree[r]{\footnotesize($\tup{\cmpr}$R)}
                                        {
                                            \prfbyaxiom{\footnotesize(Ax)}
                                            {@_i\tup{\alpha\epsilon\beta}b,@_i\tup{\gamma}k, \tagg{\cmpr}{j}{k} \vdash @_i\tup{\alpha\epsilon\beta \cmpr \gamma}, \tagg{\cmpr}{j}{k}}
                                        }
                                        {@_i\tup{\alpha\epsilon\beta}b,@_i\tup{\gamma}k, \tagg{\cmpr}{j}{k} \vdash @_i\tup{\alpha\epsilon\beta \cmpr \gamma}}
                                    }
                                    {@_i\tup{\alpha\epsilon}a, @_a\tup{\beta}b,@_i\tup{\gamma}k, \tagg{\cmpr}{j}{k} \vdash @_i\tup{\alpha\epsilon\beta \cmpr \gamma}}
                                }
                                {@_a\tup{\epsilon}a,@_i\tup{\alpha}a, @_a\tup{\beta}b,@_i\tup{\gamma}k, \tagg{\cmpr}{j}{k} \vdash @_i\tup{\alpha\epsilon\beta \cmpr \gamma}}
                            }
                            {@_a\top, @_aa,@_i\tup{\alpha}a, @_a\tup{\beta}b,@_i\tup{\gamma}k, \tagg{\cmpr}{j}{k} \vdash @_i\tup{\alpha\epsilon\beta \cmpr \gamma}}
                        }
                        {@_i\tup{\alpha}a, @_a\tup{\beta}b,@_i\tup{\gamma}k, \tagg{\cmpr}{j}{k} \vdash @_i\tup{\alpha\epsilon\beta \cmpr \gamma}}
                    }
                    {@_i\tup{\alpha\beta}j, @_i\tup{\gamma}k, \tagg{\cmpr}{j}{k} \vdash @_i\tup{\alpha\epsilon\beta \cmpr \gamma}}
                }
                {@_i\tup{\alpha\beta \cmpr \gamma} \vdash @_i\tup{\alpha\epsilon\beta \cmpr \gamma}}
            }
            {\vdash @_i(\tup{\alpha\epsilon\beta \cmpr \gamma} \liff \tup{\alpha\beta \cmpr \gamma})}
        }
    \end{center}
    \vspace{0.5cm}
    \item {\small Axiom (comp-dist).
    We prove $\vdash_{\hilbert}\tup{\alpha\beta}\varphi \liff \tup{\alpha}\tup{\beta}\varphi$ implies
    $\vdash_{\gentzen}@_i(\tup{\alpha\beta}\varphi \liff \tup{\alpha}\tup{\beta}\varphi)$.}
    \begin{center}\scalebox{\thescalefactor}
        {
            \prftree[r]{\footnotesize($\liff$R)}
            {
                \prfbyaxiom{\footnotesize(Ax)}
                {@_i\tup{\alpha\beta}\varphi \vdash @_i\tup{\alpha}\tup{\beta}\varphi}
            }
            {
                \prfbyaxiom{\footnotesize(Ax)}
                {@_i\tup{\alpha}\tup{\beta}\varphi \vdash @_i\tup{\alpha\beta}\varphi}
            }
            {\vdash @_i(\tup{\alpha\beta}\varphi \liff \tup{\alpha}\tup{\beta}\varphi)}
        }
    \end{center}
\end{itemize}

%% file: arXiv-cut.tex
\section{Cut Elimination}\label{sec:arXiv-cut}

\begin{definition}\label{def:size}
    The $\size$ of a node expression is defined by mutual recursion as:
    \begin{align*}
        \size(p) & = 1 &
            \size(\varphi \to \psi) & = 1 + \size(\varphi) + \size(\psi) &
                \size(\dowa) & = 1 \\
        \size(i) & = 1 &
            \size(@_i\varphi) & = 1 + \size(\varphi)  &
                \size(i{:}) & = 1 \\
        \size(\bot) & = 1 &
            \size(\tup{\dowa} \varphi) & = 1 + \size(\varphi) &
                \size(\varphi?) & = 1 + \size(\varphi) \\
        &&
            \size(\tup{\alpha \cmpr \beta}) & = 1 + \size(\alpha) + \size(\beta) &
                \size(\alpha \beta) & = \size(\alpha) + \size(\beta)
    \end{align*}
    The function $\size$ induces a well-founded partial order over the set of node expressions.
\end{definition}

\begin{definition}
    The \emph{height} of a derivation is the length of its longest branch (e.g., a derivation consisting of only an application of (Ax) has height 1).
    If $\Gamma \vdash \Delta$ is the end-sequent in a derivation, we will use $\Gamma \vdash^n \Delta$ to indicate that the derivation has height $n$.
    The \emph{cut height} of an application of (Cut) in a given derivation is the sum of the heights of the derivations of the premisses of the rule; i.e., if we have derivations $\Gamma \vdash^n \Delta, \cc{\varphi}$, and $\cc{\varphi}, \Gamma' \vdash^m \Delta'$, using (Cut), we obtain a derivation $\Gamma, \Gamma' \vdash^{(\max(n,m)+1)} \Delta, \Delta'$, in this case, the cut height is $n+m$.
    In such an application of (Cut), we call $\cc{\varphi}$ the \emph{active cut} expression.
\end{definition}

\begin{theorem}\label{th:cut:elimination}
    Every use of (Cut) in the derivation of a provable sequent can be eliminated.
\end{theorem}
\begin{proof}
    The proof is by induction on two measures: the size of the active cut expression, and the cut height.
    To this end, we associate with each application of (Cut) in a derivation a pair $(k,h)$, called \emph{cut complexity},  where: $k$ corresponds to the size of the active cut expression, and $h$ corresponds to the cut height.
    The induction is on the lexicographic order of the pairs $(k,h)$.
    
    \medskip\noindent
    \textbf{Base Case.}
    The base cases of the induction involve derivations in which (Cut) is applied only once, with axiom rules applied to its premisses, i.e., the premisses of (Cut) must be instances of (Ax) or ($\bot$).
    Eliminating (Cut) in such configurations is unproblematic.
    We illustrate one representative case below.
    Suppose $\aderivation[C]$ and $\aderivation[A]$ are derivations with the following structure:
    $$\scalebox{\thescalefactor}
        {
            \prftree[r,l]{\footnotesize(Cut)}{$\aderivation[C]=$}
            {
                \prfbyaxiom{\footnotesize(Ax)}
                {\Gamma \vdash \Delta, \cc{\varphi}}
            }
            {
                \prfbyaxiom{\footnotesize($\bot$)}
                {\cc{\varphi}, \Gamma' \vdash \Delta'}
            }
            {\Gamma, \Gamma'  \vdash \Delta, \Delta'}

            \qquad
            \qquad

            \prftree[r,l]{\footnotesize($\Rho$)}{$\aderivation[A]=$}
            {}
            {\Gamma, \Gamma'  \vdash \Delta, \Delta'}
        }
    $$
    $\aderivation[C]$ represents a case in which one of the premisses of (Cut) is (Ax) and the other is ($\bot$).
    This derivation can be transformed into the cut-free derivation $\aderivation[A]$ in which: $(\Rho)$ is (Ax) if $\varphi \notin \Gamma$; and it is ($\bot$) otherwise.
    The remaining cases use a similar argument.

	\medskip\noindent
    \textbf{Inductive Step.}
    We identify an application of (Cut) that is minimal according to the cut height, and eliminate it.
    We proceed by a case analysis based on the syntactic form of the active cut, and on whether the active cut is principal in either premiss of the application of (Cut).

    \medskip
    \noindent
    {\it Non-principal Cases.}
    First, let us cover the case where the active cut is not principal in the right premiss of (Cut).
    In this case, the application of (Cut) is permuted up.

    \begin{description}
        \item [Case ($\tup{\cmpr}$R).]
        Precisely, suppose that we find a derivation
        $$\scalebox{\thescalefactor}
            {
                \prftree[r]{\footnotesize(Cut)}
                {
                    \prfassumption
                    {\Gamma \vdash^{n} {\Delta}, \cc{\varphi}}
                }
                {
                    \prftree[r]{\footnotesize(${\tup{\cmpr}}$L)}
                    {
                        \prfassumption
                        {@_i\tup{\alpha}j, @_i\tup{\beta}k, \tagg{\cmpr}{j}{k}, \varphi, {\Gamma'} \vdash^m \Delta}
                    }
                    {\cc{\varphi}, @_i\tup{\alpha \cmpr \beta}, {\Gamma'} \vdash \Delta}
                }
                {@_i\tup{\alpha \cmpr \beta}, \Gamma, {\Gamma'} \vdash {\Delta}, \Delta'}
            }
        $$

        \noindent
        and that in this derivation the use of (Cut) has a cut complexity $(\size(\varphi), n + (m + 1))$, where $n+(m+1)$ is minimal.
        W.l.g., assume that $j$ and $k$ do not appear  in $\Gamma \vdash \Delta$.%
            \footnote{If $j$ or $k$ do appear in $\Gamma \vdash \Delta$, we can simply choose different nominals, and rewrite the derivation of ${{\varphi}, @_i\tup{\alpha \cmpr \beta}, {\Gamma'} \vdash \Delta}$ using the new selection of nominals.}
        We transform this derivation into:
        $$\scalebox{\thescalefactor}
            {
                \prftree[r]{\footnotesize(${\tup{\cmpr}}$L)}
                {
                    \prftree[r]{\footnotesize(Cut)}
                    {
                        \prfassumption
                        {\Gamma \vdash^{n} {\Delta}, \cc{\varphi}}
                        \quad
                    }
                    {
                        \prfassumption
                        {\cc{\varphi}, @_i\tup{\alpha}j, @_i\tup{\beta}k, \tagg{\cmpr}{j}{k}, {\Gamma'} \vdash^m \Delta}
                    }
                    {@_i\tup{\alpha}j, @_i\tup{\beta}k, \tagg{\cmpr}{j}{k}, \Gamma, {\Gamma'} \vdash {\Delta}, \Delta'}
                }
                {@_i\tup{\alpha \cmpr \beta}, \Gamma, {\Gamma'} \vdash \Delta, {\Delta'}}
            }
        $$

        \noindent
        The (Cut) in the transformed derivation has complexity $(\size(\varphi), n+m)$, so it can be eliminated by the inductive hypothesis.

    \item[Remaining Cases.]
        The remaining cases where the active cut is not principal in the right premiss follow the same strategy.
        To see why, note that all such derivations have the following general structure:
        $$\scalebox{\thescalefactor}
            {
                \prftree[r]{\footnotesize(Cut)}
                {
                    \prfassumption
                    {\Gamma \vdash^{n} \Delta,\cc{\varphi}}
                }
                {
                    \prftree[r]{\footnotesize($\Rho$)}
                    {
                        \prfassumption
                        {{\varphi},\Phi', \Gamma' \vdash^{m} \Delta',\Sigma'}
                    }
                    {\cc{\varphi},\Phi, \Gamma' \vdash \Delta',\Sigma}
                }
                {\Phi, \Gamma, \Gamma'  \vdash \Delta, \Delta', \Sigma}
            }
        $$
        In this derivation, we are considering ($\Rho$) is a single premiss rule of $\gentzen$,
        the set $\varphi,\Gamma',\Delta'$ is the context for the rule, and
        the set $\Phi',\Phi,\Sigma',\Sigma$ are the node expressions the rule acts upon.
        Again, we assume that (Cut) has complexity $(\size(\varphi),n+(m+1))$ where $n+(m+1)$ is minimal; i.e., where (Cut) is not used in ${\Gamma \vdash^{n} \Delta, \cc{\varphi}}$, nor in ${{\varphi},\Phi', \Gamma' \vdash^{m} \Delta',\Sigma'}$.
        Modulo a possible renaming of nominals, we transform this derivation into:
        $$\scalebox{\thescalefactor}
            {
                \prftree[r]{\footnotesize($\Rho$)}
                {
                    \prftree[r]{\footnotesize(Cut)}
                    {
                        \prfassumption
                        {\Gamma \vdash^{n} \Delta, \cc{\varphi}}
                    }
                    {
                        \prfassumption
                        {\cc{\varphi},\Phi',\Gamma' \vdash^{m} \Delta',\Sigma'}
                    }
                    {\Phi', \Gamma', \Gamma \vdash \Delta, \Delta', \Sigma'}
                }
                {\Phi, \Gamma', \Gamma \vdash \Delta, \Delta', \Sigma}
            }
        $$
        The use of (Cut) in the transformed derivation has complexity $(\size(\varphi), n+m)$, using the inductive hypothesis, we obtain a cut-free derivation of ${\Phi, \Gamma', \Gamma \vdash \Delta, \Delta', \Sigma}$.
    \end{description}

    \noindent
    The cases where the active cut is not principal in the left premiss is symmetric. The ($\to$L) case---the only two-premise rule in $\gentzen$---is handled similarly, and is well known in the literature.
    
    \medskip
    \noindent
    {\it Principal Cases.}
    Let us now turn our attention to derivations where the active cut is principal in both premisses.
    We examine all combinations of rules with a possibly matching active cut.

    \begin{description}
        \item [Interaction between ($\to$L) and ($\to$R).]
        Suppose that in a derivation we encounter an application of (Cut) of minimal height of the form:
        $$\scalebox{\thescalefactor}
        {
            \prftree[r]{\footnotesize(Cut)}
                {
                    \prftree[r]{\footnotesize($\to$R)}
                    {
                        \prfassumption
                        {@_i\varphi, \Gamma \vdash^n \Delta, @_i\psi}
                    }
                    {\Gamma \vdash \Delta, \cc{@_i(\varphi \to \psi)}}
                }
                {
                    \prftree[r]{\footnotesize($\to$L)}
                    {
                        \prfassumption
                        {\Gamma' \vdash^{m_1} \Delta', @_i\varphi}
                    }
                    {
                        \prfassumption
                        {@_i\psi, \Gamma' \vdash^{m_2} \Delta'}
                    }
                    {\cc{@_i(\varphi \to \psi)}, \Gamma' \vdash \Delta'}
                }
                {\Gamma, \Gamma' \vdash \Delta,\Delta'}
        }
        $$
    
        \noindent
        We transform this derivation into:
    
        $$\scalebox{\thescalefactor}
            {
                \prftree[r]{\footnotesize(Cut$_2$)}
                {
                    \prfassumption
                    {\Gamma' \vdash^{m_1} \Delta', \cc{@_i\varphi}}
                }
                {
                    \prftree[r]{\footnotesize(Cut$_1$)}
                    {
                        \prfassumption
                        {@_i\varphi, \Gamma,\vdash^{n} \Delta, \cc{@_i\psi}}
                    }
                    {
                        \prfassumption
                        {\cc{@_i\psi},\Gamma' \vdash^{m_2} \Delta' }
                    }
                    {\cc{@_i\varphi},\Gamma,\Gamma' \vdash \Delta,\Delta'}
                }
                {\Gamma, \Gamma' \vdash \Delta,\Delta'}               
            }
        $$
        
        The original use of (Cut) has complexity $(\size(@_i(\varphi \to \psi)), (n+1)+ (\max(m_{1},m_{2})+1))$.
        In the transformed derivation, (Cut$_1$) has complexity $(\size(@_i\psi), n+m_{2})$.
        In turn, the complexity of (Cut$_2$) is such that $\size(@_i\varphi) < \size(@_i(\varphi \to \psi))$.
        Both applications of (Cut) in the transformed derivation can be eliminated by the inductive hypothesis.

        \item  [Interaction between ($\tup{\dowa}$R) and ($\tup{\dowa}$L).]
        Suppose that in a derivation we encounter an application of (Cut) of minimal height of the form:
        $$\scalebox{\thescalefactor}
            {
                \prftree[r]{\footnotesize(Cut)}
                {
                    \prftree[r]{\footnotesize($\tup{\dowa}$R)}
                    {
                        \prfassumption
                        {@_i\tup{\dowa}j, \Gamma \vdash^{n} \Delta, @_i\tup{\dowa}\varphi, @_j\varphi}
                    }
                    {@_i\tup{\dowa}j,\Gamma \vdash \Delta, \cc{@_i\tup{\dowa}\varphi}}
                }
                {
                    \prftree[r]{\footnotesize($\tup{\dowa}$L)}
                    {
                        \prfassumption
                        {@_i\tup{\dowa}j, @_j\varphi, \Gamma' \vdash^{m} \Delta'}
                    }
                    {\cc{@_i\tup{\dowa}\varphi}, \Gamma' \vdash \Delta'}
                }
                {@_i\tup{\dowa}j,\Gamma, \Gamma' \vdash \Delta,\Delta'}
            }
        $$
    
        \noindent
        We transform this derivation into:
        $$\scalebox{\thescalefactor}
            {
                \prftree[r]{\footnotesize(Cut$_2$)}
                {
                    \prftree[r]{\footnotesize(Cut$_1$)}
                    {
                        \prfassumption
                        {@_i\tup{\dowa}j, \Gamma,\vdash^{n} \Delta,@_j\varphi, \cc{@_i\tup{\dowa}\varphi}}
                    }
                    {
                        \prftree[r]{\footnotesize($\tup{\dowa}$L)}
                        {
                            \prfassumption
                            {@_i\tup{\dowa}j, @_j\varphi, \Gamma' \vdash^{m} \Delta'}
                        }
                        {\cc{@_i\tup{\dowa}\varphi}, \Gamma' \vdash \Delta'}
                    }
                    {@_i\tup{\dowa}j, \Gamma,\Gamma' \vdash \Delta,\Delta', \cc{@_j\varphi}}
                }
                {
                    \prfassumption
                    {\cc{@_j\varphi},@_i\tup{\dowa}j, \Gamma' \vdash^{m} \Delta'}
                }
                {@_i\tup{\dowa}j,\Gamma, \Gamma' \vdash \Delta,\Delta'}
            }
        $$ 
        \noindent
        The original use of (Cut) has complexity $(\size(@_i\tup{\dowa}\varphi), (n+1)+(m+1))$.
        In the transformed derivation, (Cut$_1$) has complexity $(\size(@_i\tup{\dowa}\varphi), n+(m+1))$.
        In turn, the complexity of (Cut$_2$) is such that $\size(@_j\varphi) < \size(@_i\tup{\dowa}\varphi)$.
        Both applications of (Cut) in the transformed derivation can be eliminated by the inductive hypothesis.

        \item [Interaction between ($\tup{\dowa}$R) and ($\tup{\cmpr}$R).]
        Suppose that in a derivation we encounter an application of (Cut) of minimal height of the form:
        $$\scalebox{\thescalefactor}
            {
                \prftree[r]{\footnotesize(Cut)}
                {
                    \prftree[r]{\footnotesize(${\tup{\dowa}}$R)}
                    {
                        \prfassumption
                        {@_i\tup{\dowa}j, \Gamma \vdash^{n} \Delta, @_i\tup{\dowa}a, @_ja}
                    }
                    {@_i\tup{\dowa}j, \Gamma \vdash \Delta, \cc{@_i\tup{\dowa}a}}
                    \hspace{-2pt}
                }
                {
                    \prftree[r]{\footnotesize(${\tup{\cmpr}}$R)}
                    {
                        \prfassumption
                        {@_i\tup{\dowa}a, @_i\tup{\beta}b, \Gamma' \vdash^m \Delta', @_i\tup{\dowa \cmpr \beta},\tagg{\cmpr}{a}{b}}
                    }
                    {\cc{@_i\tup{\dowa}a}, @_i\tup{\beta}b, \Gamma' \vdash \Delta', @_i\tup{\dowa \cmpr \beta}}
                }
                {@_i\tup{\dowa}j, @_i\tup{\beta}b, \Gamma, \Gamma' \vdash \Delta, \Delta', @_i\tup{\dowa \cmpr \beta}}
            }
        $$
    
        \noindent
        We transform this derivation into:
    
        \begin{flushleft}
            ~~
            \scalebox{\thescalefactor}
            {
                 $\aderivation = {@_i\tup{\dowa}j, \Gamma \vdash^{n} \Delta, @_ja, \cc{@_i\tup{\dowa}a}}$
            }
        \end{flushleft}
        \vspace{-1cm}
        \begin{center}\scalebox{\thescalefactor}
            {
                \prftree[r]{\footnotesize(Cut$_2$)}
                {
                    \prftree[r]{\footnotesize(Cut$_1$)}
                    {
                        ~~
                        \prfassumption
                        {\aderivation}
                        \qquad
                    }
                    {
                        \prftree[r]{\footnotesize(${\tup{\cmpr}}$R)}
                        {
                            \prfassumption
                            {@_i\tup{\dowa}a, @_i\tup{\beta}b, \Gamma' \vdash^m \Delta', @_i\tup{\dowa \cmpr \beta},\tagg{\cmpr}{a}{b}}
                        }
                        {\cc{@_i\tup{\dowa}a}, @_i\tup{\beta}b, \Gamma' \vdash \Delta', @_i\tup{\dowa \cmpr \beta}}
                    }
                    {@_i\tup{\dowa}j, @_i\tup{\beta}b, \Gamma, \Gamma' \vdash \Delta, \Delta', @_i\tup{\dowa \cmpr \beta}, \cc{@_ja}}
                }
                {
                    \prftree[r]{\footnotesize($\text{S}_2$)}
                    {
                        \prftree[r]{\footnotesize(WL)}
                        {
                            \prftree[r]{\footnotesize(WL)}
                            {
                                \prftree[r]{\footnotesize(${\tup{\cmpr}}$R)}
                                {
                                    \prfassumption
                                    {@_i\tup{\dowa}a, @_i\tup{\beta}b, \Gamma' \vdash^m \Delta', @_i\tup{\dowa \cmpr \beta},\tagg{\cmpr}{a}{b}}
                                }
                                {{@_i\tup{\dowa}a}, @_i\tup{\beta}b, \Gamma' \vdash \Delta', @_i\tup{\dowa \cmpr \beta}}
                            }
                            {@_i\tup{\dowa}j, @_i\tup{\dowa}a, @_i\tup{\beta}b, \Gamma' \vdash \Delta', @_i\tup{\dowa \cmpr \beta}}
                        }
                        {@_ja, @_i\tup{\dowa}j, @_i\tup{\dowa}a, @_i\tup{\beta}b, \Gamma' \vdash \Delta', @_i\tup{\dowa \cmpr \beta}}
                    }
                    {\cc{@_ja}, @_i\tup{\dowa}j, @_i\tup{\beta}b, \Gamma' \vdash \Delta', @_i\tup{\dowa \cmpr \beta}}
                }
                {@_i\tup{\dowa}j, @_i\tup{\beta}b, \Gamma, \Gamma' \vdash \Delta, \Delta', @_i\tup{\dowa \cmpr \beta}}
            }
        \end{center}
    
        \noindent
        The original use of (Cut) has complexity $(\size(@_i\tup{\dowa}a), (n+1)+(m+1))$.  
        In the transformed derivation, (Cut$_1$) has complexity $(\size(@_i\tup{\dowa}a), n+(m+1))$.
        In turn, the complexity of (Cut$_2$) is such that $\size(@_ja) < \size(@_i\tup{\dowa}a)$.
        Both applications of (Cut) in the transformed derivation can be eliminated by the inductive hypothesis.
        
        \item [Interaction between ($@$L) and ($@$R).]
        Suppose that in a derivation we encounter an application of (Cut) of minimal height of the form:
        $$\scalebox{\thescalefactor}
            {
                \prftree[r]{\footnotesize(Cut)}
                {
                    \prftree[r]{\footnotesize($@$R)}
                    {
                        \prfassumption
                        {\Gamma \vdash^{n} \Delta,@_i\varphi}
                    }
                    {\Gamma \vdash \Delta, \cc{@_j@_i\varphi}}
                }
                {
                    \prftree[r]{\footnotesize($@$L)}
                    {
                        \prfassumption
                        {@_i\varphi, \Gamma' \vdash^{m} \Delta'}
                    }
                    {\cc{@_j@_i\varphi},\Gamma' \vdash \Delta'}
                }
                {\Gamma, \Gamma' \vdash \Delta, \Delta'}
            }
        $$
        
        \noindent
        We transform this derivation into:
    
        $$\scalebox{\thescalefactor}
            {
                \prftree[r]{\footnotesize(Cut)}
                {
                    \prfassumption
                    {\Gamma \vdash^{n} \Delta,\cc{@_i\varphi}}
                }
                {
                    \prfassumption
                    {\cc{@_i\varphi}, \Gamma' \vdash^{m} \Delta'} 
                }
                {\Gamma, \Gamma' \vdash \Delta, \Delta'}
            }
        $$
    
        \noindent
        The original use of (Cut) has complexity $(\size(@_j@_i\varphi), (n+1)+(m+1))$.
        In the transformed derivation, the new application of (Cut) has complexity $(\size(@_i\varphi), n+m)$. 
        And can thus be eliminated by the inductive hypothesis.

        \item [Interaction between ($\tup{\cmpr}$R) and ($\tup{\cmpr}$L).]
        Suppose that in a derivation we encounter an application of (Cut) of minimal height of the form:
        $$\scalebox{\thescalefactor}
            {                   
                \prftree[r]{\footnotesize(Cut)}
                {
                    \prftree[r]{\footnotesize(${\tup{\cmpr}}$R)}
                    {
                        \prfassumption
                        {@_i\tup{\alpha}j, @_i\tup{\beta}k, \Gamma \vdash^{n} \Delta, @_i\tup{\alpha \cmpr \beta},\tagg{\cmpr}{j}{k}}
                    }
                    {@_i\tup{\alpha}j, @_i\tup{\beta}k, \Gamma \vdash \Delta, \cc{@_i\tup{\alpha \cmpr \beta}}}
                }
                {
                    \prftree[r]{\footnotesize(${\tup{\cmpr}}$L)}
                    {@_i\tup{\alpha}j, @_i\tup{\beta}k, \tagg{\cmpr}{j}{k}, \Gamma' \vdash^m \Delta'}
                    {\cc{@_i\tup{\alpha \cmpr \beta}}, \Gamma' \vdash \Delta'}
                }
                {@_i\tup{\alpha}j, @_i\tup{\beta}k, \Gamma, \Gamma' \vdash \Delta, \Delta'}
            }
        $$
        
        \noindent
        We transform this derivation into:
        $$\scalebox{\thescalefactor}
            {
                \prftree[r]{\footnotesize(Cut$_2$)}
                {
                    \prftree[r]{\footnotesize(Cut$_1$)}
                    {
                        \prfassumption
                        {@_i\tup{\alpha}j, @_i\tup{\beta}k, \Gamma \vdash^{n} \Delta, \tagg{\cmpr}{j}{k}, \cc{@_i\tup{\alpha \cmpr \beta}}}
                    }
                    {
                        \hspace*{-.3cm}
                        \prftree[r]{\footnotesize(${\tup{\cmpr}}$L)}
                        {@_i\tup{\alpha}j, @_i\tup{\beta}k, \tagg{\cmpr}{j}{k}, \Gamma' \vdash^m \Delta'}
                        {\cc{@_i\tup{\alpha \cmpr \beta}}, \Gamma' \vdash \Delta'}
                    }
                    {@_i\tup{\alpha}j, @_i\tup{\beta}k, \Gamma, \Gamma' \vdash \Delta, \Delta', \cc{\tagg{\cmpr}{j}{k}}}
                }
                {
                    \hspace*{-1.3cm}
                    \prfassumption
                    {\cc{\tagg{\cmpr}{j}{k}}, @_i\tup{\alpha}j, @_i\tup{\beta}j, \Gamma' \vdash^m \Delta'}
                }
                {@_i\tup{\alpha}j, @_i\tup{\beta}k, \Gamma, \Gamma' \vdash \Delta, \Delta'}
            }
        $$
        
        \noindent
        The original use of (Cut) has complexity $(\size(@_i\tup{\alpha \cmpr \beta}), (n+1)+(m+1))$.
        In the transformed derivation, (Cut$_1$) has complexity $(\size(@_i\tup{\alpha \cmpr \beta}), n+(m+1))$.
        In turn, the complexity of (Cut$_2$) is such that $\size(\tagg{\cmpr}{j}{k}) < \size(@_i\tup{\alpha \cmpr \beta})$.
        Both applications of (Cut) in the transformed derivation can be eliminated by the inductive hypothesis.

        \item [Interaction between (NEqL) and (NEqR).] 
        Suppose that in a derivation we encounter an application of (Cut) of minimal height of the form:
        $$\scalebox{\thescalefactor}
            {
                \prftree[r]{\footnotesize(Cut)}
                {
                    \prftree[r]{\footnotesize(NEqR)}
                    {
                        \prfassumption
                        {\tagg{=}{i}{j},\Gamma \vdash^{n} \Delta}
                    }
                    {\Gamma \vdash \Delta, \cc{\tagg{\neq}{i}{j}}}
                }
                {
                    \prftree[r]{\footnotesize(NEqL)}
                    {
                        \prfassumption
                        {\Gamma' \vdash^{m} \Delta',\tagg{=}{i}{j}}
                    }
                    {\cc{\tagg{\neq}{i}{j}},\Gamma' \vdash \Delta'}
                }
                {\Gamma, \Gamma' \vdash \Delta, \Delta'}
            }
        $$
    
        \noindent
        We transform this derivation into:
    
        $$\scalebox{\thescalefactor}
            {
                \prftree[r]{\footnotesize(Cut)}
                {
                    \prfassumption
                    {\Gamma' \vdash^{m} \Delta',\cc{\tagg{=}{i}{j}}}
                }
                {
                    \prfassumption
                    {\cc{\tagg{=}{i}{j}}, \Gamma \vdash^{n} \Delta}
                }
                {\Gamma, \Gamma' \vdash \Delta, \Delta'}
            }
        $$

        \noindent
        The original use of (Cut) has complexity $(\size(\tagg{\neq}{i}{j}), (n+1)+(m+1))$.
        The use of (Cut) in the transformed derivation has complexity $(\size(\tagg{=}{i}{j}), n+m)$.
        This second use can be eliminated applying the inductive hypothesis.

        \item [Interaction between ($\tup{\dowa}$R) and (S$_1$).] 
        Suppose that in a derivation we encounter an application of (Cut) of minimal height of the form:
        $$\scalebox{\thescalefactor}
            {
                \prftree[r]{\footnotesize(Cut)}
                {
                    \prftree[r]{\footnotesize(${\tup{\dowa}}$R)}
                    {
                        \prfassumption
                        {@_i\tup{\dowa}j, \Gamma \vdash^{n} \Delta, @_i\tup{\dowa}a, @_ja}
                    }
                    {@_i\tup{\dowa}j, \Gamma \vdash \Delta, \cc{@_i\tup{\dowa}a}}
                    \qquad
                }
                {
                    \prftree[r]{\footnotesize($\text{S}_1$)}
                    {
                        \prfassumption
                        {@_k\tup{\dowa}a, @_i\tup{\dowa}a, @_ik, \Gamma' \vdash^{m} \Delta'}
                    }
                    {\cc{@_i\tup{\dowa}a}, @_ik, \Gamma' \vdash \Delta'}
                }
                {@_i\tup{\dowa}j, @_ik, \Gamma, \Gamma' \vdash \Delta, \Delta'}
            }
        $$
        
        \noindent
        We transform this derivation into:

        $$\scalebox{\thescalefactor}
            {
                \prftree[r]{\footnotesize(Cut$_2$)}
                {
                    \prftree[r]{\footnotesize(Cut$_1$)}
                    {
                        \prfassumption
                        {@_i\tup{\dowa}j, \Gamma \vdash^{n} \Delta, @_ja, \cc{@_i\tup{\dowa}a}}
                    }
                    {
                        \prftree[r]{\footnotesize($\text{S}_1$)}
                        {
                            \prfassumption
                            {@_k\tup{\dowa}a, @_i\tup{\dowa}a, @_ik, \Gamma' \vdash^{m} \Delta'}
                        }
                        {\cc{@_i\tup{\dowa}a}, @_ik, \Gamma' \vdash \Delta'}
                    }
                    {@_i\tup{\dowa}j, @_ik, \Gamma, \Gamma' \vdash \Delta, \Delta', \cc{@_ja}}
                }
                {
                   \prftree[r]{\footnotesize(S$_2$)}
                    {
                        \prftree[r]{\footnotesize(WL)}
                        {
                            \prftree[r]{\footnotesize(WL)}
                            {
                                \prfassumption
                                {@_i\tup{\dowa}a, @_ik, \Gamma' \vdash^{(m{+}1)} \Delta'}
                            }
                            {@_i\tup{\dowa}j, @_i\tup{\dowa}a, @_ik, \Gamma' \vdash \Delta'}
                        }
                        {@_ja, @_i\tup{\dowa}j, @_i\tup{\dowa}a, @_ik, \Gamma' \vdash \Delta'}
                    }
                    {\cc{@_ja}, @_i\tup{\dowa}j, @_ik, \Gamma' \vdash \Delta'}
                }
                {@_i\tup{\dowa}j, @_ik, \Gamma, \Gamma' \vdash \Delta, \Delta'}
            }
        $$
        
        \noindent
        The original use of (Cut) has complexity $(\size(@_i\tup{\dowa}a), (n+1)+(m+1))$.  
        In the transformed derivation, (Cut$_1$) has complexity $(\size(@_i\tup{\dowa}a), n+(m+1))$.
        In turn, the complexity of (Cut$_2$) is such that $\size(@_ja) < \size(@_i\tup{\dowa}a)$.
        Both applications of (Cut) in the transformed derivation can be eliminated using the inductive hypothesis.
        
        \item [Interaction between ($\tup{\dowa}$R) and (S$_2$).] 
        Suppose that in a derivation we encounter an application of (Cut) of minimal height of the form:
    
        $$\scalebox{\thescalefactor}
            {
                \prftree[r]{\footnotesize(Cut)}
                {
                    \prftree[r]{\footnotesize($\tup{\dowa}$R)}
                    {
                        \prfassumption
                        {@_i\tup{\dowa}j,\Gamma \vdash^{n} \Delta,@_i\tup{\dowa}a, @_ja}
                    }
                    {@_i\tup{\dowa}j,\Gamma \vdash \Delta,\cc{@_i\tup{\dowa}a}}
                }
                {
                    \prftree[r]{\footnotesize(S$_2$)}
                    {
                        \prfassumption
                        {@_i\tup{\dowa}k, @_i\tup{\dowa}a, @_ak, \Gamma' \vdash^{m} \Delta'}
                    }
                    {\cc{@_i\tup{\dowa}a},@_ak, \Gamma' \vdash \Delta'}
                }
                {@_ak,@_i\tup{\dowa}j,\Gamma, \Gamma' \vdash \Delta, \Delta'}
            }
        $$
        
        \noindent
        We transform this derivation into:

        $$\scalebox{\thescalefactor}
            {
                \prftree[r]{\footnotesize(Cut$_2$)}
                {
                    \prftree[r]{\footnotesize(Cut$_1$)}
                    {
                        \prfassumption
                        {@_i\tup{\dowa}j,\Gamma \vdash^{n} \Delta,@_ja, \cc{@_i\tup{\dowa}a}}
                    }
                    {
                        \prfassumption
                        {\cc{@_i\tup{\dowa}a},@_ak, \Gamma' \vdash^{(m+1)} \Delta'}
                    }
                    {@_ak,@_i\tup{\dowa}j,\Gamma, \Gamma' \vdash \Delta, \Delta', \cc{@_ja}}
                }
                {
                    \prftree[r]{\footnotesize(S$_2$)}
                    {
                        \prftree[r]{\footnotesize(S$_2$,WL)}
                        {
                            \prfassumption
                            {@_i\tup{\dowa}a, @_i\tup{\dowa}k, @_ak,\Gamma' \vdash^{m} \Delta'}
                        }
                        {@_i\tup{\dowa}a, @_ja, @_ak, @_i\tup{\dowa}j, \Gamma' \vdash \Delta'}
                    }
                    {\cc{@_ja}, @_ak, @_i\tup{\dowa}j, \Gamma' \vdash \Delta'}
                }
                {@_ak,@_i\tup{\dowa}j,\Gamma, \Gamma' \vdash \Delta, \Delta'}
            }
        $$

        \noindent
        The original use of (Cut) has complexity $(\size(@_i\tup{\dowa}a), (n+1)+(m+1))$.  
        In the transformed derivation, (Cut$_1$) has complexity $(\size(@_i\tup{\dowa}a), n+(m+1))$.
        In turn, the complexity of (Cut$_2$) is such that $\size(@_ja) < \size(@_i\tup{\dowa}a)$.
        Both uses of (Cut) can be eliminated using the inductive hypothesis.
    \end{description}
\end{proof}